\documentclass{article}
\usepackage{preprint,times}
\usepackage{amsmath,amsfonts,amssymb,amsthm}
\usepackage{booktabs}
\usepackage{graphicx}
\usepackage{multirow}
\usepackage{placeins}
\usepackage{hyperref}
\usepackage{url}

\usepackage{amsmath,amsfonts,bm}

\def\eqref#1{equation~\ref{#1}}

\def\1{\bm{1}}

\DeclareMathAlphabet{\mathsfit}{\encodingdefault}{\sfdefault}{m}{sl}
\SetMathAlphabet{\mathsfit}{bold}{\encodingdefault}{\sfdefault}{bx}{n}

\newtheorem{proposition}{Proposition}
\newtheorem{corollary}{Corollary}

\graphicspath{{figures/}}

\title{LLM-as-a-Judge Version Shift:\\ The Instrument Changed, Not the Standard}

\author{Guijia Zhang \\
\normalfont The Hong Kong University \\
of Science and Technology
\And
Harry Yang \\
\normalfont The Hong Kong University \\
of Science and Technology
}

\newcommand{\dver}{\Delta_{\mathrm{ver}}}

\preprintcopy

\begin{document}

\maketitle
\lhead{Preprint}

\begin{abstract}
Automated research pipelines score proposals with an LLM reviewer and replace that reviewer when a newer checkpoint ships. A common planning assumption is that the new judge is harsher. We test that assumption in a preregistered, weight-frozen design. Three same-family open-weight transitions, Qwen, Gemma, and Mistral, score a frozen pool of ML proposals. The assumption does not hold. Qwen~$3.5\!\to\!3.6$ moves toward leniency on anchors and on the pool, and that shift survives FDR correction. Gemma~$3\!\to\!4$ is strict on constructed anchors and attenuates on the pool, so its direction is only suggestive. Mistral-Small~$3.1\!\to\!3.2$ also has a more lenient mean. That mean is the wrong summary: $93\%$ of pool scores land on a single value, and $59\%$ of the top-$30\%$ selection set changes. Giving both versions the same low-level facts, or the same defect certificates, does not close the remaining gap. A matched placebo attributes the small fact effect to packet format rather than content. Global monotone recalibration repairs the mean and almost none of the object-level dispersion. The outcomes fall into three failure classes of a measurement function: a moved origin, a change in object order, and a loss of resolution. Each class requires a different action. A stratified anchor audit of about $200$ calls classifies the class before the new judge scores anything that matters. Thresholds frozen from that protocol issue the matching diagnosis on two later API pairs. We release six preregistration locks and $14{,}882$ judge calls.
\end{abstract}

\section{Introduction}
\label{sec:intro}

Pipelines that use an LLM as a reviewer eventually replace that reviewer. The prompts, the score threshold, and the rest of the stack stay put; only the weights change. The change is treated as an upgrade. In a frozen replication of one such swap, Mistral-Small~$3.1\!\to\!3.2$, $59\%$ of the top-$30\%$ selection set turns over and the new judge assigns $93\%$ of proposals the same score. The pipeline still returns numbers. It is no longer measuring.

The usual explanation is that newer models are stricter. That story now informs planning, from leaderboard maintenance to autonomous research systems that iterate proposals against a judge \citep{alpharesearch2025,lu2024aiscientist}. Recency does not imply a direction. A mean shift is also the wrong summary when the instrument has stopped discriminating.

\begin{figure}[t]
\centering
\includegraphics[width=\textwidth]{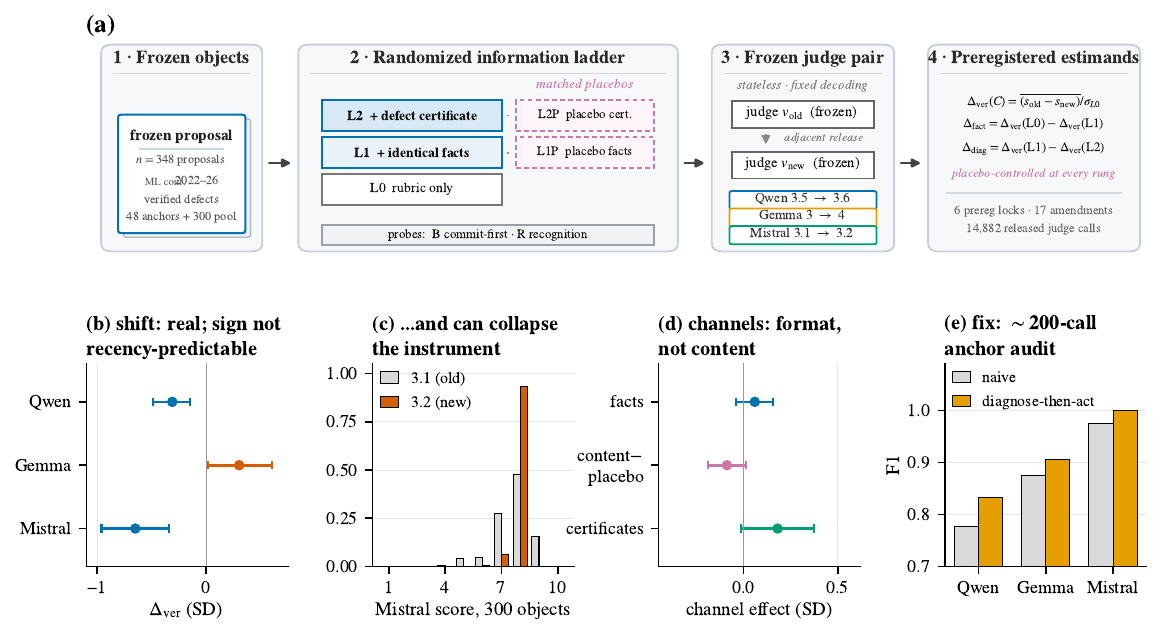}
\vspace{-4pt}
\caption{\textbf{Overview.} Panel a is the design. Frozen proposals with reviewer-verified defects are scored by weight-frozen adjacent judge versions in three families, under a randomized information ladder with matched placebos at every rung. Panels b to e are the findings. Qwen moves toward leniency and replicates on the pool. Gemma is strict on anchors and attenuates off-anchor, as Table~\ref{tab:pool} reports. Mistral~$3.1\!\to\!3.2$ concentrates $93\%$ of pool scores on a single value. Fact effects track placebos, so the active ingredient is format rather than content. A $\sim$200-call anchor audit classifies the failure before deployment, as Table~\ref{tab:protocol} records.}
\label{fig:overview}
\vspace{-9pt}
\end{figure}

The closest prior work already treats evaluator replacement as a measurement problem \citep{judgechanges2026}. That audit compares capability tiers and served APIs on pairwise judgment datasets and shows that upgrades are not interchangeable. It does not freeze adjacent same-family revision hashes, does not randomize the information shown to both versions, and does not separate a mean shift from a rank reversal or a collapse of resolution. Those distinctions decide the fix. Recalibrating a location shift is cheap. Recalibrating a reordering cannot recover the pairs that traded places. Recalibrating a collapsed score channel is vacuous.

We treat a judge version as a measurement function. When the weights change, three things can move independently: the origin of the scale, the direction in which criteria are combined, and the resolution of the score map. Each movement has a distinct observable and a different repair. A reliability audit that reports only a mean cannot tell them apart. Section~\ref{sec:idea} states the model. The experiments then do two jobs that a mean-only audit cannot. They freeze the objects and the weights, so the score difference is a property of the instrument. They give both versions the same information, so a residual gap cannot be explained as a difference in what the two models already know. Figure~\ref{fig:overview} summarizes the design and the outcomes.

\textbf{Design.} We freeze three same-family open-weight transitions and a pool of ML proposals with reviewer-verified defects \citep{soundnessbench2026}. Each judge scores each object under randomized, stateless conditions on an information ladder: rubric only, identical frozen facts, and a minimal defect certificate, with a matched placebo at every rung. Every stage was locked before any judge output of that stage was generated \citep{nosek2018}. Section~\ref{sec:design} gives the protocol.

\textbf{Findings.}
\begin{enumerate}
\item \textbf{Recency predicts neither sign nor quality.} Qwen~$3.5\!\to\!3.6$ moves toward leniency by about $0.3$~SD on both anchors and the pool. Mistral~$3.1\!\to\!3.2$ collapses the score channel: $93\%$ of pool scores sit on one value. Gemma~$3\!\to\!4$ is strict on constructed anchors and attenuates on the pool, so an anchor-only audit can overstate a shift. Three transitions are not a population of drift. Section~\ref{sec:transitions}.
\item \textbf{The residual gap is not a knowledge gap.} Identical facts close at most $0.14$~SD, and a placebo attributes even that to packet format. Certificates help one chain and fail a preregistered replication. A global map can repair a mean and absorbs under $8\%$ of object-level dispersion. On Gemma the pool gap flips sign across quality tertiles. That pattern is what a global severity parameter cannot represent. Section~\ref{sec:decomposition}.
\item \textbf{A small in-domain audit can classify the class.} About $200$ stratified anchor calls ask, in order, whether resolution collapsed, whether ranks moved beyond same-version noise, or whether only the origin moved. The actions are to withhold the new judge, to rescore, or to calibrate. The thresholds are a protocol proposal, formalized after stage-2 data existed. On two later API pairs the frozen rule issues the matching diagnosis. Section~\ref{sec:deployment}.
\end{enumerate}
Recognition and revision-loop audits are in the appendix. Mistral~$3.2$ collapses on proposals and does not on code review: reliability is a property of the operating point, not of the model name.

The contributions are a preregistered measurement that falsifies the harshness rule in both direction and quality; a placebo-controlled information ladder that bounds the usual contamination and context objections; a model that separates origin, order, and resolution failures; and a stratified anchor-audit protocol whose preregistered nulls are reported as exclusion bounds \citep{schuirmann1987,lakens2017}.

\section{Related work}
\label{sec:related}

\textbf{Judge replacement.} \citet{judgechanges2026} is the closest prior audit. They document evaluator-replacement ambiguity on four pairwise judgment datasets, comparing Qwen3 scaling from $1.7$B to $32$B with adjacent MiniMax API releases. Upgrades are not interchangeable. Only the $1.7$B to $4$B step is a robust accuracy gain, and adjacent API releases are not. They recommend dataset slices, bias probes, and protocol logs.

Four design choices separate that audit from this one. They compare capability tiers and served APIs. We freeze revision hashes on adjacent same-family pairs, including two same-size patches. They hold the item fixed and vary the judge. We also randomize the information shown to both versions, with a placebo at every rung, so a residual gap is not a gap in context. They do not report a scoring-channel collapse with intact rationales. An accuracy delta would miss Mistral~$3.2$. They do not split remediations by failure class.

Score instability over time and across versions is also documented in scientific ideation and production judges \citep{scijudgedrift2025,shelflife2025,diagaudit2026}. Related work adds anchor-based attribution \citep{whodrifted2026} and psychometric reliability audits \citep{relwithoutval2026,zheng2023judging}. Those studies establish that drift exists. They do not say what the gap is made of, whether supplying the same facts closes it, or which repair is even possible.

\textbf{Contamination.} Relatedness between a generator and an evaluator can leak preference \citep{li2025preferenceleakage}. Object-in-training-data contamination of judges has been a stated worry without a measurement. Our recognition probe with title verification measures an observable trace and bounds its contribution. Recognition is not membership evidence, and we do not treat it as such.

\textbf{Cue interventions.} Placebo-controlled cue studies exist for judge rationalization \citep{faithfulfabricated2026}. Noise schedules \citep{nrc2026}, reference conflicts \citep{judgingreference2026}, and injected text \citep{promptinjection2025} all move scores. Our packets are randomized, non-evaluative, and provenance-carrying, with matched placebos, because the estimand is a channel effect rather than attack success. Three of our own effects would have been over-claimed without those arms.

\textbf{Calibration and anchors.} Many-facet models fit severity differences after the fact \citep{irtjudge2026,rasch2026}. Fixed-parameter calibration extends benchmarks over time \citep{habba2026growingpains}. Anchor choice itself moves rankings \citep{donyehiya2026mediocrity}. Corrector choice is an anchor-budget question \citep{debiastransport2026}. The simplification those tools share is that a two-parameter map of old scores onto new ones is enough. That is what we test and reject off-anchor. Fitting a mean is not recovering a ranking.

\textbf{Overoptimization.} Imperfect proxies can be farmed \citep{gao2023scaling,moreconvincing2026,evalnotopt2026,legalqa2026}. At a stateless, six-round operating point we do not detect moderate-to-large judge-specific residue. Appendix~\ref{app:revision} reports the bounds. Shared-context optimizers remain outside the tested envelope \citep{pan2024spontaneous}.

\section{A measurement model of version shift}
\label{sec:idea}

Write a judge version $v$ as a map from an object $x$ to a score
\begin{equation}
s_v(x) = g_v\!\left(w_v^{\top} z(x) + b_v\right),
\label{eq:response}
\end{equation}
where $z(x)$ collects latent criterion-level assessments of $x$, including soundness, evidence, and novelty; $w_v$ weights those criteria; $b_v$ is an offset; and $g_v$ is strictly increasing. Replacing the weights can move any of the three parameters. Each movement has a distinct observable, and each requires a different action. Table~\ref{tab:classes} records the correspondence. A reliability audit that reports only a mean cannot tell the three classes apart.

If only $b_v$ moves, every object shifts by the same amount on the latent scale. The mean score changes. Ranks at a fixed threshold may move near the cut, but the order of objects is preserved. A monotone recalibration of old scores recovers the new scores. This is an \emph{origin} change.

If the direction of $w_v$ changes, objects are reweighted and some pairs reverse order. No function of the old score alone can restore those pairs: two objects with the same old score can carry different new scores. This is an \emph{order} change. Appendix~\ref{app:theory} records a sufficient condition under which a global monotone map $s_{\mathrm{new}} = h(s_{\mathrm{old}})$ exists. We do not treat that condition as identified. $z(x)$ is unobserved, so the data classify the kind of change, not which criterion weights moved.

If the slope of $g_v$ flattens over the realized range, discrimination collapses. Upstream channels can remain responsive. The model can still list defects and write object-specific rationales while emitting a constant scalar. Recalibrating a constant is vacuous. This is a \emph{resolution} failure. The instrument has to be withheld, or the operating point changed.

These three readings are an idealization. What the design identifies is the class of change, through the signatures in Table~\ref{tab:classes}. A second question is causal rather than descriptive. A gap that is only a difference in what the two versions already know should close when both are given the same facts. A gap that is a change in the response function should not. The information ladder tests that alternative. It starts from a bare rubric, then adds identical frozen facts, then a defect certificate, each with a matched placebo. If the residual after clamping is still object-dependent, the swap changed the instrument, not the evidence.

Section~\ref{sec:transitions} reports which class each transition instantiates. Section~\ref{sec:decomposition} tests the knowledge alternative. Section~\ref{sec:deployment} asks whether a cheap in-domain audit can detect the class before the new judge scores anything that matters.

\begin{table}[h]
\caption{Three failure classes of a version change. The design classifies the class through these signatures. It does not identify the latent weights.}
\label{tab:classes}
\centering
\small
\begin{tabular}{llll}
\toprule
Class & What moves & Signature & Required action \\
\midrule
Origin & $b_v$ & location shift, order intact & calibrate \\
Order & $w_v$ & rank reversals, residual dispersion & rescore \\
Resolution & $g_v$ & score concentration, intact rationales & withhold \\
\bottomrule
\end{tabular}
\end{table}

\section{Experimental design}
\label{sec:design}

The design has two jobs: classify each transition against Table~\ref{tab:classes}, and test whether the gap is a difference in what the two versions already know.

\textbf{Judges.} Three same-family open-weight transitions are frozen by revision hash: Qwen~3.5-27B$\to$3.6-27B, Gemma~3-27B$\to$4-31B, and Mistral-Small~3.1-24B$\to$3.2-24B. The Qwen and Mistral pairs are same-size patch revisions. Gemma is a generation jump with a size change. We treat all three as the kind of replacement a pipeline actually performs, not as exchangeable effect sizes.

Decoding is fixed at temperature $0.6$ and top-$p$ $0.95$, with per-call CRC seeds. That setting matches a common deployment choice rather than greedy decoding. Prompt-repeat reliability is quadratic-weighted $\kappa = 0.86$--$0.93$. A later check shows the Mistral collapse is stable across temperature. A pool-scale second pass separates sampling noise from version change. Section~\ref{sec:deployment} and Appendix~\ref{app:api4th} report those checks.

\textbf{Objects.} The pool has 348 frozen proposal-style objects derived from ML conference submissions from 2022 to 2026, with reviewer soundness scores and verified risk factors \citep{soundnessbench2026}. They fall in three disjoint sets. The first is 48 anchors, of which 12 are clean, 12 are flawed, and 24 are natural. The second is an 80-object intervention pool. The third is a 220-object expansion pool spanning five submission years. Full pool means the 300 non-anchor objects. The 62 masked-title variants are rewrites of these objects, not additional samples. All bootstraps cluster at the object level. Appendix~\ref{app:accounting} gives the call counts.

\textbf{The ladder.} For each object we issue stateless randomized calls. Condition B is commit-first defect listing with no score, following \citet{moreconvincing2026}. Condition R is a title-verified recognition probe. Condition L0 is the rubric only. Condition L1 adds a frozen atomic-fact packet, built by a neutral third family and identical for both judges. Condition L1P is a matched placebo packet. Condition L2 adds a minimal defect certificate distilled from reviewer risk factors and leakage-filtered. Condition L2P is a donor-mismatched certificate matched for length, style, and authority cues. Qwen and Gemma run the full ladder. Mistral was added at stage~2 as a third-family replication and runs L0, R, B, L2, and L2P, with no fact conditions.

For condition $C$, let $\dver(C) = \overline{(s_{\mathrm{old}} - s_{\mathrm{new}})}\,/\,\sigma_{L0}$ denote the standardized version gap under $C$, an object-clustered mean, where $\sigma_{L0}$ is the transition's pooled L0 SD. Positive $\dver$ means the new judge is stricter. The preregistered estimands are $\dver \equiv \dver(\mathrm{L0})$, $\Delta_{\mathrm{fact}} = \dver(\mathrm{L0}) - \dver(\mathrm{L1})$, and $\Delta_{\mathrm{diag}} = \dver(\mathrm{L1}) - \dver(\mathrm{L2})$.

\textbf{Preregistration.} Six locks, one per stage, freeze configs, prompts, code hashes, kill gates, and analysis scripts before any judge output of that stage. Seventeen amendments are logged in Appendix~\ref{app:amendments}. Each was documented before unblinding the affected endpoint. Amendments extend scope or add replications; they do not redefine a confirmatory endpoint after unblinding. Every test is registered as confirmatory or exploratory in Table~\ref{tab:fullregistry}. Negative results carry TOST exclusion bounds and 80\%-power minimum detectable effects. Appendix~\ref{app:equivalence} reports those bounds.

\FloatBarrier
\section{Three transitions: two shifts and a collapse}
\label{sec:transitions}

Table~\ref{tab:pool} and Figure~\ref{fig:transitions} classify the three transitions against the model of Section~\ref{sec:idea}. Throughout, shift refers to the location estimand $\dver$.

\begin{figure}[ht]
\centering
\includegraphics[width=0.78\textwidth]{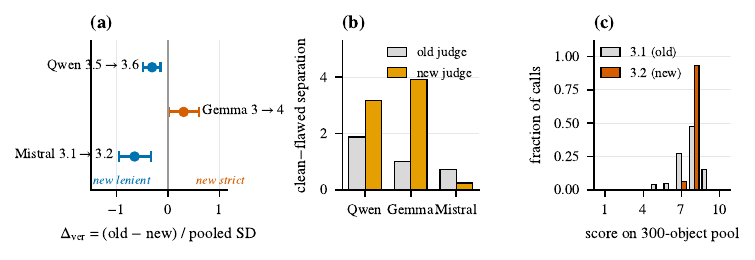}
\vspace{-4pt}
\caption{\textbf{Version replacement re-instruments the measurement.} Panel a reports location shifts with cluster-bootstrap 95\% CIs. The Qwen shift toward leniency survives FDR correction. Gemma's direction toward strictness is suggestive. The Mistral mean also moves, but panel c is the finding that matters on that chain. Panel b reports clean--flawed discrimination, which improves in two transitions and collapses in the third. Panel c shows the Mistral~3.2 collapse on the 300-object full pool, where $93\%$ of scores land on a single value.}
\label{fig:transitions}
\vspace{-9pt}
\end{figure}

\begin{table}[ht]
\caption{Transition endpoints, anchors versus full pool, in SD units with object-clustered bootstrap 95\% CIs. Positive $\dver$ means the new judge is stricter. Pool is the 220-object expansion pool scored once per judge at L0. For Mistral the pool is 300 objects. Anchor and pool gaps are each standardized by their own pooled L0 SD. Sep.\ is clean--flawed separation. Rank columns quantify transition-induced rank instability. Section~\ref{sec:deployment} returns to those columns.}
\label{tab:pool}
\centering
\scriptsize
\setlength{\tabcolsep}{3pt}
\begin{tabular}{lrrrrrr}
\toprule
Transition & $\dver$ anchors [95\% CI] & sep. & $\dver$ pool [95\% CI] & $\tau$ & top-30\% flip & ties \\
\midrule
Qwen 3.5$\to$3.6 & $-0.310$ $[-0.487, -0.146]$ & $1.88 \to 3.17$ & $-0.332$ $[-0.432, -0.236]$ & $0.613$ & $37\%$ & $18\%$ \\
Gemma 3$\to$4 & $+0.303$ $[+0.014, +0.607]$ & $1.00 \to 3.92$ & $+0.110$ $[-0.034, +0.258]$ & $0.506$ & $39\%$ & $26\%$ \\
Mistral 3.1$\to$3.2 & $-0.650$ $[-0.961, -0.339]$ & $0.73 \to \mathbf{0.24}$ & $-0.450$ $[-0.616, -0.290]$ & $0.175$ & $59\%$ & $86\%$ \\
\bottomrule
\end{tabular}
\end{table}

Qwen is predominantly an origin change. It moves toward leniency by $0.31$~SD on anchors and $0.33$~SD on the pool. Both intervals exclude zero.

Mistral is a resolution failure. On the pool, $93\%$ of scores sit on one value, clean--flawed separation falls from $0.73$ to $0.24$, and $86\%$ of pairs are tied. The SD-standardized mean is $-0.65$ on anchors and $-0.45$ on the pool. That number is a fragile summary of a collapsed channel, so we lead with concentration and rank statistics rather than with leniency. The collapse is a finding about instrument risk. That chain is excluded from the calibration claims of Section~\ref{sec:deployment}. Only its degradation flag enters the protocol table.

Gemma is strict on the constructed anchors, $+0.30$~SD, and much smaller on the pool, $+0.11$. The pool interval includes zero. We treat the direction as suggestive.

The harshness claim is universal, so the two FDR-surviving shifts toward leniency already refute it. We do not claim that leniency is typical.

\textbf{Anchor estimates are audit statistics, not pool estimates.} Anchors are a constructed set, so every headline gap is also reported on the expansion pool, and for Mistral on the 300-object full pool, sampled without regard to defect status. Table~\ref{tab:pool} and Appendix~\ref{app:pool} give the numbers. Qwen's magnitude agrees. Gemma's attenuation is the reason we call its direction suggestive, and a warning that an anchor-only audit can overstate a shift.

\textbf{Multiplicity.} As a sensitivity analysis we apply Benjamini--Hochberg FDR correction across the eleven scalar confirmatory endpoints. Appendix~\ref{app:robustness} lists them. The Qwen location shift survives, $p < 10^{-3}$. The Mistral mean shift also survives, but we do not read it as a leniency finding. Rank statistics on the anchors show $46\%$ exact ties, rank-biserial $0.74$, and Wilcoxon $p = 4 \times 10^{-4}$. Gemma's direction toward strictness, $p \approx 0.045$, does not survive FDR, matching the pool attenuation. No other conclusion changes.

\textbf{The collapse is in the scalar score, not in the harness or the reading.} Three audits localize Mistral~3.2. Appendix~\ref{app:robustness} reports them in full. The harness is intact: $0$ of $252$ pool calls fail to parse, output lengths are ordinary, and there is no repetition pathology. Comprehension is intact: justifications are object-specific, with mean pairwise 3-gram Jaccard $0.047$. The model reads each proposal, writes a tailored rationale, and then emits the same score. Discrimination is also intact on the commit-first B probe. Mistral~3.2's defect counts separate clean from flawed anchors by $0.55$~SD, Cliff's $\delta = 0.30$, numerically better than its predecessor at $-0.27$~SD. The mapping to a scalar is what broke.

A preregistered perturbation matrix, Table~\ref{tab:e1}, completes the attribution. The collapse is stable across decoding temperature. It partly recovers under a semantically equivalent prompt paraphrase, with SD rising from $0.31$ to $0.53$, still at least $42\%$ below the predecessor. The same paraphrase leaves the predecessor unchanged. The saturation binds to prompt surface form. Only an in-domain audit will show which operating point still works.

The collapse is also domain-specific. On $48$ code-review objects the same pair moves toward leniency by $-0.79$~SD with discrimination intact. Appendix~\ref{app:domain2} reports that study.

The sign of a working instrument's mean shift tracks where the \emph{old} judge was miscalibrated. Old Qwen under-scored clean objects. Old Gemma over-scored flawed ones. Reading that sign requires anchored ground truth, which is why Section~\ref{sec:deployment} scores a frozen in-domain set before the new judge is trusted.

\section{What does not close the gap}
\label{sec:decomposition}

\begin{figure}[t]
\centering
\includegraphics[width=0.6\textwidth]{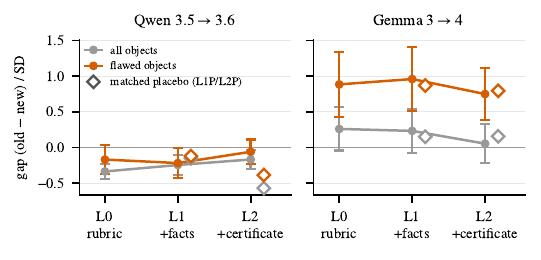}
\vspace{-4pt}
\caption{\textbf{The information ladder.} Version gap, old minus new, in pooled-SD units, with means and 95\% CIs, under L0, L1, and L2. L0 is the rubric. L1 adds identical frozen facts. L2 adds a defect certificate. Open diamonds mark the matched placebos at each rung. Facts barely move either chain's gap, and the placebos track the real interventions. The placebo-controlled contrast attributes what little the packets do to format, not content, as Table~\ref{tab:e2} reports. Certificates shrink the gap in Qwen but fail a preregistered Mistral replication, as Table~\ref{tab:certs} reports.}
\label{fig:ladder}
\vspace{-9pt}
\end{figure}

Section~\ref{sec:idea} left a causal alternative: the gap could be a difference in what the two versions already know. The information ladder tests that alternative. If identical facts or identical certificates restored the old scores, the swap would be an evidence gap rather than an instrument change.

\textbf{Non-evaluative fact supply does not close the gap.} Clamping identical frozen fact packets moves the pooled gap by $\Delta_{\mathrm{fact}} = +0.060$~SD. The interval is $[-0.039, +0.158]$, on $128$ paired objects on the two fact-bearing chains, sign-aligned. Equivalence testing excludes effects of $0.143$~SD or larger. Identical low-level context is not a fix.

\textbf{The active ingredient is format, not content.} Irrelevant placebo facts already perturb Qwen anchors as strongly as real ones, $0.33$ versus $0.04$~SD, so L1 confounds content with packet form. A preregistered placebo-controlled extension separates the channels. L1P is placed on every L1-covered object. The sample is $128$ paired objects. Table~\ref{tab:e2} reports the contrast. The content-specific effect is null, pooled $-0.088$~SD, interval $[-0.190, +0.012]$. TOST excludes effects of $0.174$~SD or larger. A placebo packet alone moves the gap toward zero by $0.11$--$0.17$~SD, significantly on Qwen and numerically more than the real facts do.

The packets are model-extracted, leakage-filtered atomic facts. They lexically cover about $5\%$ of reviewer-verified risk factors, against $29\%$ for certificates. Appendix~\ref{app:artifacts} details the construction. What the data rule out is that low-level facts or auto-generated certificates generally close the gap. Expert-curated evidence remains untested.

\textbf{Certificates do not generally bypass.} Table~\ref{tab:certs} carries the numbers. The preregistered primary is pooled $\Delta_{\mathrm{diag}}$ on flawed objects. It is inconclusive under our own bar, and the placebo gate did not trigger. An exploratory paired contrast that cancels object noise came out significant. A preregistered third-family replication on Mistral returned null. That replication used the same 33 flawed objects, with sign, scale, and a one-sided test frozen in advance, and the manipulation check passed. The certificate channel is chain-specific at best. Mismatched pseudo-certificates still depress clean-object scores in every judge, three times more strongly in old Qwen, $-0.95$ versus $-0.30$ raw. Version replacement changes the response to authoritative-sounding text, not just severity.

\textbf{Prior exposure is the weakest finding.} Title-verified recognition differs across versions. New Qwen recognizes $29\%$ of objects, against $12\%$ for old Qwen. Gemma~3 claims recognition it cannot verify $36\%$ of the time, so self-report is not an audit. Verified recognition establishes prior exposure or reconstruction, not membership. Recognized objects carry an excess gap of $0.155$~SD under L0, shrinking to $0.04$~SD under fact-clamping. The strata are observational and imbalanced. Recognized objects are older and higher-scored, with SMD $0.49$--$0.73$. Appendix~\ref{app:balance} reports the balance table.

Three checks are consistent with an information advantage and no more. The excess persists under placebo packets. Gap dispersion is comparable across strata. A masked-title difference-in-differences is directionally consistent but underpowered, $+0.095$, and the interval contains zero.

\textbf{Global maps repair the mean, not the objects.} After clamping, about $77\%$ of both chains' gaps survive. We treat what follows as residual attribution, not mechanism identification. The preregistered two-parameter model is $s_{\mathrm{new}} \approx \gamma\, s_{\mathrm{old}} + \delta$, fitted on 24 anchors and required to explain at least two thirds of the off-anchor mean gap. It failed on both chains. Qwen overshoots, explained fraction $-0.15$. Gemma extrapolates, $\gamma = 2.5$, explained fraction $-12.5$. Part of the Gemma failure is an anchor artifact. The old judge compresses anchors into a two-point range, so the slope is unidentified \citep{donyehiya2026mediocrity}.

Refitting on random subsets of the 220-object pool, 200 draws at $k = 24$, $40$, and $80$, shows the failure is not only about anchors. On Qwen the mean becomes correctable, with median explained fraction $0.76$--$0.81$ for $k \geq 40$. On Gemma the two-thirds gate still fails in $77\%$ of draws at $k = 80$. On both chains a linear map fitted on the full pool absorbs less than $8\%$ of object-level gap dispersion. Appendix~\ref{app:refit} reports the refit.

The Gemma pool gap itself flips sign across reviewer overall-score tertiles. The gap is $+0.80$~SD on weak objects, $-0.20$ in the middle, and $-0.43$ on strong ones. Qwen is uniformly lenient, at $-0.23$, $-0.37$, and $-0.41$. A global severity parameter cannot represent a sign flip of that kind. Appendix~\ref{app:pool} reports the subgroup table.

These are the signatures of Section~\ref{sec:idea}. Qwen is an origin-dominant shift whose mean is repairable off-anchor and whose objects are not. Gemma is an order change. The tertile reversal is the signature that a global map cannot represent. Mistral remains a resolution failure. The scalar channel collapsed while rationales and defect lists did not. Rank reversals exceed the same-version noise floor on every chain, as Section~\ref{sec:deployment} reports. Where the Gemma gap is reliable, split-half $r = 0.81$, and observable features predict $46\%$ of its variance. Compressing length or stripping rhetoric leaves the gap unchanged. Appendix~\ref{app:structural} has those audits. The change binds to content properties that survive both edits. Estimating $\Delta w$ would require per-criterion interventions with matched placebo edits. Section~\ref{sec:limitations} returns to that experiment.

\FloatBarrier
\section{A pre-deployment compatibility test}
\label{sec:deployment}

The three classes in Section~\ref{sec:idea} need three different actions. Table~\ref{tab:protocol} is the protocol we propose. The questions are asked in that order. A single drift statistic conflates failures that require opposite remediations.

\begin{table}[h]
\caption{Pre-deployment compatibility test on stratified in-domain anchors, aligned with Table~\ref{tab:classes}. Thresholds were set by convention and formalized after stage-2 data existed. They are a protocol proposal. Appendix~\ref{app:loto} reports a leave-one-transition-out check. That check shows insensitivity, not external validity. The two API pairs test the diagnosis, not the realized payoff of every prescribed action.}
\label{tab:protocol}
\centering
\scriptsize
\setlength{\tabcolsep}{3.5pt}
\begin{tabular}{llll}
\toprule
Class & Ask & Trigger & Action \\
\midrule
Resolution & Score channel collapsed? & modal share $>0.5$ or sep.\ $<0.30$ & Do not deploy \\
Order & Ranks moved? & $\tau$ below both test--retest floors & Rescore or freeze old \\
Origin & Origin only? & boundary disagreement $>10\%$ & Calibrate on anchors \\
\bottomrule
\end{tabular}
\end{table}

\begin{figure}[h]
\centering
\includegraphics[width=0.48\textwidth]{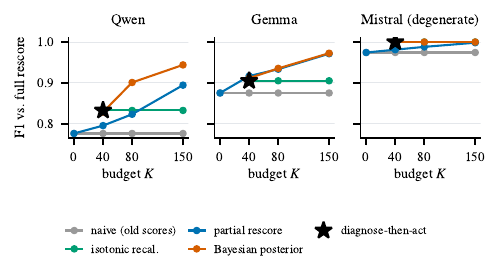}
\vspace{-4pt}
\caption{\textbf{Equal-budget controller at a version transition.} F1 is measured against the full-rescore policy as a function of new-judge call budget $K$. Stars mark the diagnose-then-act rule of Table~\ref{tab:protocol} at the anchor budget. The Mistral panel is degenerate and is excluded from calibration claims.}
\label{fig:controller}
\vspace{-9pt}
\end{figure}

Separation alone is not a reliable flag at the boundary. Anchor resampling leaves Qwen and Gemma untouched but puts collapsed Mistral near the $0.30$ threshold, with bootstrap flag rate $51\%$. Score concentration sits far from any boundary: $93\%$ of scores on a single value and $86\%$ tied pairs. That is why Table~\ref{tab:protocol} pairs the two flags.

\textbf{Scalar maps and rank pipelines.} Fitted on 40 anchors and evaluated on the non-anchor pool, ordinal probit improves Qwen MAE from $0.833$ to $0.811$ and CRPS from $0.65$ to $0.57$. On Gemma every mapping underperforms identity, as the tertile reversal predicts. Monotone maps leave within-chain Kendall $\tau$ unchanged, so a rank-based pipeline can skip scalar recalibration of a \emph{fixed} judge. The exemption does not cover the transition. Replacing the judge reorders objects. Table~\ref{tab:pool} reports $\tau = 0.61$, $0.51$, and $0.18$, and $37$--$59\%$ of the top-$30\%$ set flips.

A second L0 pass over every pool object measures the sampling floor. Same-version test--retest $\tau$ is $0.68$--$0.82$, and $0.34$ for pre-collapse Mistral. Cross-version $\tau$ falls below both endpoints on all three chains. Version change adds $+12$, $+35$, and $+14$ flip points beyond the same-version mean for Qwen, Gemma, and Mistral. Appendix~\ref{app:api4th} reports the floor. Rank pipelines still have to audit rank stability, boundary crossings, and discrimination at each swap.

\textbf{Equal-budget migration.} Figure~\ref{fig:controller} selects the new-judge top-$30\%$ under a fixed rescoring budget. On Qwen, isotonic anchor recalibration lifts F1 from $0.777$ to $0.833$ and cuts regret by $40\%$. Regret is the shortfall in mean reviewer soundness against the best same-size selection. Beyond the anchor budget a Bayesian posterior layer reaches F1 $0.901$--$0.944$ at budgets of $80$ to $150$. Isotonic keeps the better regret, $0.32$ versus $0.46$. Fidelity to the incoming judge, measured by F1, and fidelity to external quality, measured by regret, are separate objectives. Best arm refers to F1. On Gemma the naive policy already has regret near zero, because the shift misses the selection boundary. The diagnose-then-act rule in Table~\ref{tab:protocol} matches the post-hoc best F1 arm on the two non-degenerate chains and raises the degradation flag on Mistral. We count only Qwen and Gemma as calibration evidence. Mistral's F1 comparison is trivial.

\textbf{Independent API pairs.} With thresholds frozen, pinned snapshots GPT-5.4 $\to$ GPT-5.5 scored the 48 anchors under the identical L0 prompt. The shift is $0.45$~SD toward strictness, interval $[0.32, 0.57]$. Separation stays intact, $3.52$ to $3.11$. Boundary disagreement is $12.5\%$. The rule triggers calibration and withholds the degradation flag. A pool-scale rescore of 660 calls, with endpoints frozen first, confirms the diagnosis out of sample. Pool $\dver = 0.73$~SD, interval $[0.64, 0.82]$, and disagreement is $33\%$. The payoff is pipeline-dependent. Anchor-fitted isotonic calibration cuts boundary disagreement to $19\%$. A monotone map cannot alter top-$k$ selection, so the frozen fidelity endpoint is null by construction. Appendix~\ref{app:api4th} reports that check.

A second blind pair from a second provider, DeepSeek flash $\to$ pro, is a tier replacement rather than a same-size patch. It returns a $-0.42$~SD shift toward leniency, interval $[-0.60, -0.25]$, with improved separation from $1.78$ to $2.12$. The rule again withholds the degradation flag and triggers calibration. DeepSeek was not rescored at pool scale, so that pair tests diagnosis only.

\textbf{Scope.} The thresholds are an in-simulation demonstration. The two API pairs are the first independent tests of the diagnosis. Pinned snapshots are not weight-frozen, and decoding is provider-default. Those estimates are descriptive. Stress tests of the audit itself, reported in Appendix~\ref{app:structural}, put the minimum viable size near 24 stratified anchors for healthy instruments. Near-threshold trigger verdicts should be escalated to a pool-scale check.

\FloatBarrier
\section{Limitations and outlook}
\label{sec:limitations}

The primary dataset is proposal-style text. A 48-object code-review replication shows that sign and magnitude differ by domain. Appendix~\ref{app:domain2} reports that study. Three preregistered transitions plus two descriptive API pairs refute universal claims and demonstrate the audit. They are not a population estimate of drift. Cross-family replacement on the same pool is at least as disruptive and was not a confirmatory endpoint. Appendix~\ref{app:structural} reports that check.

Criterion assessments are latent, so the geometry reading rests on observed signatures rather than identified weights. Anchor selection remains a live failure mode \citep{donyehiya2026mediocrity}. Pinned API snapshots are not weight-frozen \citep{diagaudit2026}. Fact packets and certificates are model-generated and non-evaluative by design. Appendix~\ref{app:artifacts} states the construction. The revision null is one operating point. Appendix~\ref{app:revision} states the bound. Recognition is prior exposure, not membership. The protocol issues the matching diagnosis on both API pairs and remains unvalidated under cross-family replacement in the field.

Three experiments would close the remaining gaps. Identifying $\Delta w$ needs per-criterion interventions with matched placebo edits. A field test of Table~\ref{tab:protocol} needs a swap that the authors do not choose, including a cross-family replacement. Expert-curated evidence would test whether the fact-channel null is an artifact of model-extracted packets. Those studies are the next paper.

\section{Conclusion}
\label{sec:conclusion}

Replacing an LLM judge changes the instrument, not a scalar severity. Origin, criterion direction, and resolution can move independently, and recency does not pick the sign. Identical facts do not restore the old scores, and a global map cannot recover the objects that traded places. A compatibility test on frozen in-domain anchors classifies the class before the new judge scores anything that matters: withhold when resolution is gone, rescore when ranks have moved, and calibrate only when the origin moved. The instrument changed. The standard does not have to.

\subsubsection*{Reproducibility statement}
All 14{,}882 judge calls (13{,}646 open-weight, 1{,}236 API) are released as JSONL with model identifiers, per-call seeds, and raw outputs, together with six preregistration locks (one per stage: 0a, 0b, 1a, 1b, 2, 3; each a SHA-256 freeze of configs, prompts, and analysis code written before any judge output of that stage), seventeen amendments (A1, B1, B2, C1, D1, E1, E2, F1--F3, G1--G3, G4D, G5--G7; Appendix~\ref{app:amendments}), and single-script regeneration of every table cell. Open-weight judges are pinned by revision hash with deterministic decoding given the released seeds; the API pairs are pinned by dated snapshot or tier ID, with all raw responses released. Auxiliary figures (evaporation scatter, per-stratum rolling MAE, controller regret curves, the B-probe defect-count comparison) ship with the artifacts.

\subsubsection*{Ethics statement}
This study evaluates frozen model snapshots on public research proposals; no human subjects were involved. We discuss dual-use of the masking manipulation (it could be misread as a contamination-evasion recipe) in Appendix~\ref{app:amendments}; our use is diagnostic, with matched controls.

\subsubsection*{AI use statement}
In this work, we used generative AI tools to build experimental infrastructure (orchestration and analysis code), to generate the intervention artifacts that are themselves objects of study (fact packets, certificates, masked titles, and revisions, all frozen and released), and to assist drafting. We did not use generative AI to design the study, select endpoints, or interpret results; retrieval-augmented ideation is not applicable to this work. We have reviewed all AI-assisted work: analysis code was verified against the frozen call logs, every reported number is regenerable by a single released script, and all experimental design decisions, preregistration gates, and claims were checked against the frozen artifacts by the authors. We take responsibility for the final content of this work, including text, claims, and artifacts produced with the aid of generative AI.

\bibliography{references}
\bibliographystyle{preprint}

\appendix

\noindent\textit{Appendix map.} The measurement model and the calibration proposition: Appendix~\ref{app:theory}. Collapse forensics and the temperature/prompt matrix: Appendix~\ref{app:forensics}. Calibration refit and the dispersion ceiling: Appendix~\ref{app:refit}. Recognition-stratum balance: Appendix~\ref{app:balance}. Audit protocol, leave-one-transition-out, and API pairs: Appendix~\ref{app:loto} and Appendix~\ref{app:api4th}. Code-review replication: Appendix~\ref{app:domain2}. Revision-loop nulls: Appendix~\ref{app:revision}. Amendments do not redefine a confirmatory endpoint after unblinding: Appendix~\ref{app:amendments}.

\section{Full test registry and preregistration locks}
\label{app:registry}

Table~\ref{tab:fullregistry} lists every test the project ran, its lock, its confirmatory or exploratory status, and its outcome, including the failures and nulls. It is the single place to check what was promised before data and what was delivered after; the amendments log (Appendix~\ref{app:amendments}) covers every post-lock change.

\begin{table}[h]
\caption{Complete confirmatory/exploratory registry. Lock files: \texttt{prereg\_lock.json} (0a), \texttt{prereg\_lock\_0b.json}, \texttt{prereg\_lock\_1a.json}, \texttt{prereg\_lock\_1b.json}, \texttt{prereg\_lock\_2.json} (contains amendment D1), \texttt{prereg\_lock\_3.json} (amendments E1, E2). Each lock stores SHA-256 hashes of configs, prompts, and analysis code, written before any judge output of that stage.}
\label{tab:fullregistry}
\centering
\scriptsize
\begin{tabular}{p{4.6cm}llp{4.6cm}}
\toprule
Test & Lock & Status & Result \\
\midrule
$\dver$ Qwen (48 anchors, L0$\times$3, cluster bootstrap) & 0a & confirmatory & $-0.310$ $[-0.487,-0.146]$ \\
$\dver$ Gemma (amendment A1) & 0a & confirmatory & $+0.303$ $[+0.014,+0.607]$ \\
$\dver$ Mistral (third transition) & 2 & confirmatory & $-0.650$ $[-0.961,-0.339]$ \\
Reliability (prompt-repeat, quadratic $\kappa$) & 0a & manipulation check & $0.86$--$0.93$ \\
$\Delta_{\mathrm{fact}}$ pooled, paired ($n{=}128$) & 0b & confirmatory & $+0.060$ $[-0.039,+0.158]$; TOST $\pm 0.143$ \\
Recognition excess gap, L0 vs.\ L1 (obs.\ strata) & 0b & confirmatory & $0.155 \to 0.04$ SD \\
L1P placebo-fact sensitivity (anchors, $n{=}24$) & 0b & warning-level & placebo $\approx$ real-fact effect \\
Pooled $\Delta_{\mathrm{diag}}$ flawed ($k{=}4$, inv.-var.\ pooled) & 1a & confirmatory primary & $+0.181$ $[-0.014,+0.375]$, $p{\approx}0.069$ \\
R3 kill gate (placebo ratio $+$ significance) & 1a & confirmatory gate & not triggered ($\Delta_{\mathrm{plac}}=-0.019$) \\
Paired $\Delta_{\mathrm{diag}}-\Delta_{\mathrm{plac}}$ ($n{=}55$) & --- & exploratory & $+0.214$ $[+0.059,+0.368]$ \\
Mistral replication of paired contrast ($n{=}33$, one-sided) & 2/D1 & confirmatory & $-0.082$ $[-0.493,+0.329]$, $p{=}0.66$; no collapse (SD $0.68$) \\
Two-parameter threshold model, $2/3$ gate & 1a & confirmatory & failed: expl.\ frac.\ $-0.15$ (Qwen), $-12.5$ (Gemma) \\
Compliance asymmetry (clean $+$ mismatched cert.) & 1a & secondary & J\_old $-0.95$, J\_current $-0.30$, G\_old/G\_new $-0.27$ \\
Masking DiD (new-only vs.\ neither, paired) & 1a & confirmatory & $+0.095$; manip.\ check $29\% \to 3\%$ \\
Revision $\rho$(sensitivity, evaporation), $n{=}40$ & 1b & confirmatory primary & $0.156$, $p{=}0.34$; TOST $\pm 0.412$ \\
In-loop gain r1$\to$r6 (single-call trajectory) & 1b & manipulation check & $+0.15$ $[-0.23,+0.55]$; fresh-seed $+0.304$ $[+0.096,+0.527]$ \\
Judge-loop $-$ generic control (paired, $n{=}20$) & 1b & confirmatory & $+0.062$ $[-0.174,+0.302]$; TOST $\pm 0.262$ \\
Winner's curse (in-loop $-$ fresh, $n{=}40$) & 1b & confirmatory & $-0.075$ $[-0.300,+0.150]$; TOST $\pm 0.275$ \\
Amendment C1 (Gemma EXP baselines, 114 calls) & 1b & infrastructure & estimation set $21 \to 40$ objects \\
Rolling prediction (6 maps $\times$ 3 chains) & 2 & confirmatory & Table~\ref{tab:rolling} \\
Controller simulation ($+$ meta-policy) & 2 & simulation & Figure~\ref{fig:controller} \\
E1 perturbation matrix (prompt $\times$ temperature, 816 calls) & 3 & confirmatory & partial recovery under frozen rule; Table~\ref{tab:e1} \\
E2 placebo-controlled fact effect (480 calls) & 3 & confirmatory & content null: $-0.088$ $[-0.190,+0.012]$; format real (Table~\ref{tab:e2}) \\
BH multiplicity sensitivity (11 endpoints) & --- & sensitivity & Table~\ref{tab:bh}; Gemma direction suggestive \\
\bottomrule
\end{tabular}
\end{table}

\section{Equivalence tests}
\label{app:equivalence}

For each preregistered null, we report the TOST margin (the smallest $\Delta$ such that the 90\% bootstrap CI lies within $[-\Delta, +\Delta]$) and the minimum detectable effect at 80\% power ($\mathrm{MDE}_{80} = 2.8 \times \mathrm{SE}_{\mathrm{boot}}$) \citep{schuirmann1987,lakens2017}. An interpretive caveat we adopt throughout: because these margins are derived from the data rather than pre-specified as smallest effect sizes of interest (SESOI), they are \emph{exclusion bounds}, statements that effects at least this large are ruled out at $\alpha{=}0.05$, not equivalence declarations in the strict Lakens sense. Main-text phrasing (``TOST excludes $\geq x$'') follows this reading; where a bound exceeds what would be a defensible SESOI (e.g.\ the Mistral replication, $\pm 0.427$), the endpoint should be read as inconclusive rather than null. The bounds are also per-endpoint at the tested operating point (e.g.\ revision depth 6): they do not constrain the slow accumulation of sub-bound residues across repeated pipeline generations.

\begin{table}[h]
\caption{Equivalence analysis of null and inconclusive endpoints (20{,}000 bootstrap draws).}
\label{tab:equivalence}
\centering
\footnotesize
\begin{tabular}{lrrrrr}
\toprule
Endpoint & $n$ & Estimate & 95\% CI & TOST margin & $\mathrm{MDE}_{80}$ \\
\midrule
Certificate $-$ placebo (pooled) & 55 & $+0.214$ SD & $[+0.059, +0.368]$ & $\pm 0.344$ & $0.222$ \\
Mistral replication (D1) & 33 & $-0.082$ SD & $[-0.493, +0.329]$ & $\pm 0.427$ & $0.588$ \\
$\rho$(sensitivity, evaporation) & 40 & $+0.156$ & $[-0.170, +0.454]$ & $\pm 0.412$ & $0.430$ \\
Judge-loop $-$ generic (held-out) & 20 & $+0.062$ SD & $[-0.174, +0.302]$ & $\pm 0.262$ & $0.338$ \\
Winner's curse & 40 & $-0.075$ pts & $[-0.300, +0.150]$ & $\pm 0.275$ & $0.322$ \\
$\Delta_{\mathrm{fact}}$ pooled & 128 & $+0.060$ SD & $[-0.039, +0.158]$ & $\pm 0.143$ & $0.142$ \\
Content-specific fact effect (E2) & 128 & $-0.088$ SD & $[-0.190, +0.012]$ & $\pm 0.174$ & $0.144$ \\
\bottomrule
\end{tabular}
\end{table}

\section{Certificate-channel, threshold-model, and compliance detail}
\label{app:threshold}

This appendix carries the full numbers behind three \S\ref{sec:decomposition} claims: the certificate channel does not survive replication (Table~\ref{tab:certs}), the anchor-fitted two-parameter map fails its preregistered gate (Table~\ref{tab:threshold}), and judges comply asymmetrically with authoritative-sounding but mismatched text (Table~\ref{tab:compliance}).

\begin{table}[h]
\caption{Certificate-channel endpoints (SD units unless noted; discussion in \S\ref{sec:decomposition}). The exploratory contrast that a replication killed is the central methodological exhibit.}
\label{tab:certs}
\centering
\footnotesize
\setlength{\tabcolsep}{4pt}
\begin{tabular}{llll}
\toprule
Endpoint & Status & Estimate [95\% CI] & Verdict \\
\midrule
Pooled $\Delta_{\mathrm{diag}}$, flawed & confirm. & $+0.181$ $[-0.014, +0.375]$, $p{=}.069$ & inconclusive \\
Placebo gate (pooled $\Delta_{\mathrm{plac}}$) & confirm. & $-0.019$ & not triggered \\
Paired $\Delta_{\mathrm{diag}} - \Delta_{\mathrm{plac}}$ & exploratory & $+0.214$ $[+0.059, +0.368]$ & needs repl. \\
Third-family replication (D1) & confirm. & $-0.082$ $[-0.493, +0.329]$, $p{=}.66$ & \textbf{null} \\
Per-chain $\Delta_{\mathrm{diag}}$ & descriptive & Qwen $+0.33$; Gemma $+0.05$; Mistral $-0.08$ & chain-specific \\
\bottomrule
\end{tabular}
\end{table}

\begin{table}[h]
\caption{Two-parameter recalibration fitted on 24 anchors, evaluated on non-anchors. Explained fraction $= 1 - |{\rm corrected\ gap}|/|{\rm raw\ gap}|$; preregistered gate required $\geq 2/3$ on both chains.}
\label{tab:threshold}
\centering
\small
\begin{tabular}{llrrrr}
\toprule
Chain & Model & $\gamma$ & $\delta$ & corrected gap (SD) & explained fraction \\
\midrule
Qwen & linear & $0.950$ & $0.894$ & $+0.201$ (raw $-0.175$) & $-0.154$ \\
Qwen & ordinal probit & $0.949$ & $0.898$ & $+0.203$ & $-0.164$ \\
Gemma & linear & $2.500$ & $-11.67$ & $-0.743$ (raw $-0.055$) & $-12.50$ \\
Gemma & ordinal probit & $2.499$ & $-11.66$ & $-0.743$ & $-12.50$ \\
\bottomrule
\end{tabular}
\end{table}

\begin{table}[h]
\caption{Compliance asymmetry: raw-score shift on \emph{clean} objects when a mismatched pseudo-certificate (L2P) is attached, vs.\ an audit-scope checklist certificate (L2). Old Qwen complies $3\times$ more than new Qwen.}
\label{tab:compliance}
\centering
\small
\begin{tabular}{lrr}
\toprule
Judge & L2P shift (pts) & L2 checklist shift (pts) \\
\midrule
Qwen 3.5 (old) & $-0.95$ $[-1.46, -0.49]$ & $+0.18$ $[-0.02, +0.39]$ \\
Qwen 3.6 (new) & $-0.30$ $[-0.60, -0.02]$ & $+0.15$ $[-0.02, +0.34]$ \\
Gemma 3 (old) & $-0.27$ $[-0.46, -0.08]$ & $+0.04$ $[-0.12, +0.19]$ \\
Gemma 4 (new) & $-0.27$ $[-0.54, 0.00]$ & $+0.27$ $[+0.08, +0.46]$ \\
\bottomrule
\end{tabular}
\end{table}

\section{Rolling prediction detail}
\label{app:rolling}

Table~\ref{tab:rolling} reports the full mapping-by-chain MAE grid behind the rolling-prediction claim of \S\ref{sec:deployment}: calibration helps modestly on Qwen, hurts on Gemma (anchor-to-pool shift exceeds what any within-chain map absorbs), and is meaningless on the degenerate Mistral column.

\begin{table}[h]
\caption{MAE of old$\to$new score mappings fitted on 40 anchors, evaluated on the full non-anchor pool (objects scored by both versions; Qwen $n{=}258$). Kendall $\tau$ is identical across strictly monotone maps within a chain (Qwen $0.629$ on this evaluation set; the $0.613$ of Table~\ref{tab:pool} is the same statistic restricted to the 220-object expansion pool). The Mistral column is an artifact of degenerate new-judge output (93\% single-value) and is excluded from claims.}
\label{tab:rolling}
\centering
\small
\begin{tabular}{lrrr}
\toprule
Mapping & Qwen & Gemma & Mistral$^{\dagger}$ \\
\midrule
identity (no calibration) & $0.833$ & $\mathbf{1.011}$ & $0.631$ \\
global shift & $0.872$ & $1.195$ & $0.796$ \\
linear (anchor) & $0.811$ & $1.203$ & $0.164$ \\
ordinal probit & $\mathbf{0.811}$ & $1.207$ & $0.286$ \\
isotonic & $0.853$ & $1.367$ & $0.095$ \\
quantile map & $0.872$ & $1.329$ & $0.062$ \\
\bottomrule
\end{tabular}
\end{table}

Recognition-stratum MAE differs by $<0.03$ within every chain and mapping; by-year MAE varies by $<0.11$ between 2022 and 2026 objects (largest for post-cutoff 2026). CRPS improves more than MAE under probit calibration on Qwen ($0.65 \to 0.57$).

\section{Revision loop}
\label{app:revision}

If judges have idiosyncratic thresholds, iterative revision against a fixed judge can farm those idiosyncrasies and collapse on the next version \citep{gao2023scaling,legalqa2026}. We preregistered a fair version of that attack (Figure~\ref{fig:nulls}): 40 objects stratified by measured threshold sensitivity (empirical-Bayes--shrunk residual gap after the global map), six rounds of revise-with-judge-feedback (neutral third-family reviser, text-only), a 20-object generic-revision control arm, fresh-seed rescoring to control winner's curse, and held-out evaluation by the other family.

\begin{figure}[t]
\centering
\includegraphics[width=0.62\textwidth]{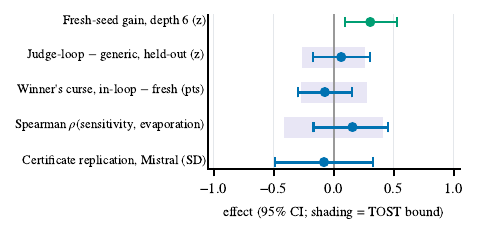}
\caption{Revision loop at a stateless, six-round operating point. Effects with 95\% bootstrap CIs; shaded bands mark the exclusion boundary for preregistered nulls. The positive control (fresh-seed gain) passes. Every exploitation signature is null within the stated bounds.}
\label{fig:nulls}
\end{figure}

Gains are real (fresh-seed depth-6 gain $+0.304$~SD $[+0.096, +0.527]$) but not judge-specific (judge-feedback beats generic revision on held-out judges by $+0.062$~SD $[-0.174, +0.302]$; TOST excludes $\geq 0.262$), not threshold-linked ($\rho$(sensitivity, evaporation) $= 0.156$, $p{=}0.34$, $n{=}40$; TOST excludes $|\rho| \geq 0.41$), and not seed-farmed (winner's curse $-0.075$ points $[-0.300, +0.150]$). Evaporation is the non-transferable share of gain.

The judges are stateless, which removes the shared conversational context that drives in-context reward hacking \citep{pan2024spontaneous}. The claim is scoped: under a stateless, six-round, mild-reviser operating point, we detect no moderate-to-large judge-specific overfitting. The bound excludes judge-specific gains $\geq 0.26$~SD (80\%-power MDE $0.34$). Effects of $0.10$--$0.25$~SD remain compatible. $n{=}40$ powers the correlation test only for $|\rho| \geq 0.43$. Optimizers that share context, run deeper, or accumulate sub-bound residues across generations are outside the tested envelope. An interim $n{=}9$ control contrast looked significant and evaporated at the preregistered full sample.

\section{Robustness analyses}
\label{app:robustness}

All analyses in this appendix use only frozen, released calls; scripts are in \texttt{paper/analysis/rebuttal\_analyses.py} and \texttt{paper/analysis/reviewer2\_analyses.py}, plus the Stage-3 lock (\texttt{prereg\_lock\_3.json}). Analyses added after unblinding (this appendix) are labeled robustness, not confirmatory.

\subsection{Mistral 3.2 collapse: forensics and perturbation matrix (E1)}
\label{app:forensics}

\textbf{Harness audit} (252 pool L0 calls): 0 parse failures; output length mean 460 chars (p10--p90: 364--573); no truncation or repetition degeneration. \textbf{Comprehension audit}: justifications are object-specific, with mean pairwise 3-gram Jaccard overlap $0.047$ (p90 $= 0.094$); the model reads each proposal and writes a tailored rationale before emitting the modal score (93.3\% single value). \textbf{Discrimination audit} (commit-first B probe, 48 anchors): defect counts separate clean from flawed anchors for Mistral~3.2 ($8.0$ vs.\ $8.75$; $0.55$ SD; Cliff's $\delta = +0.30$) but not for Mistral~3.1 ($12.5$ vs.\ $11.7$; $-0.27$ SD): the successor's defect-listing channel is, if anything, sharper. \textbf{Robust estimands} for $\dver$ (48 anchors): median per-object gap $0$ (46\% exact ties); among non-ties the new judge scores higher with rank-biserial $0.74$; Wilcoxon $p = 4\times 10^{-4}$. The mean-based $-0.65$ SD is directionally confirmed by rank statistics but numerically dominated by the non-tied minority, which is why we report both.

\textbf{Perturbation matrix (E1).} Table~\ref{tab:e1} reports the preregistered endpoints. The collapse is fully robust to decoding temperature (between-object SD $0.31$--$0.35$ at $t \in \{0.0, 0.6, 1.0\}$ under the original prompt, vs.\ predecessor reference $0.94$) but sensitive to prompt surface form: a semantically equivalent paraphrase of the same rubric restores part of the spread (SD $0.45$--$0.55$; modal share $0.90 \to 0.50$--$0.73$; clean--flawed separation up to $1.64$) while remaining $\geq 42\%$ below the predecessor in every cell. The same paraphrase leaves the predecessor essentially unchanged (SD $0.87$ vs.\ $0.94$), so the sensitivity is specific to the new model. Under the frozen decision rule this is \emph{partial recovery}: neither a harness artifact nor a configuration-independent property of the weights, but a saturation of the score mapping that is anchored to prompt surface form. This sharpens rather than weakens the deployment lesson: the instrument is defined jointly by weights and prompt, and a version transition can break one operating point while leaving a paraphrase partially functional, which is undetectable without an anchor audit.

\begin{table}[h]
\caption{E1 perturbation matrix: Mistral~3.2 (M\textsubscript{new}) between-object SD of median L0 scores on 48 anchors, by prompt variant and temperature; predecessor (M\textsubscript{old}) reference cells included. 0 parse failures in all 816 new calls.}
\label{tab:e1}
\centering
\small
\begin{tabular}{llrrr}
\toprule
Judge & Cell & SD & modal share & separation \\
\midrule
M\textsubscript{new} & orig, $t{=}0.0$ & $0.31$ & $0.90$ & $0.44$ \\
M\textsubscript{new} & orig, $t{=}0.6$ & $0.31$ & $0.90$ & $0.24$ \\
M\textsubscript{new} & orig, $t{=}1.0$ & $0.35$ & $0.88$ & $0.57$ \\
M\textsubscript{new} & para, $t{=}0.0$ & $0.55$ & $0.50$ & $1.64$ \\
M\textsubscript{new} & para, $t{=}0.6$ & $0.53$ & $0.65$ & $0.92$ \\
M\textsubscript{new} & para, $t{=}1.0$ & $0.45$ & $0.73$ & $1.13$ \\
\midrule
M\textsubscript{old} & orig, $t{=}0.6$ (ref.) & $0.94$ & $0.42$ & $0.73$ \\
M\textsubscript{old} & para, $t{=}0.6$ & $0.87$ & $0.60$ & $1.17$ \\
\bottomrule
\end{tabular}
\end{table}

\subsection{Threshold-model refit on representative pools}
\label{app:refit}

\begin{table}[h]
\caption{Two-parameter model refitted on random subsets of the 220-object expansion pool (200 draws per cell), evaluated on held-out pool objects. ``Pass'' $=$ share of draws meeting the preregistered $2/3$ explained-fraction gate. Anchor diagnostics explain the Stage-1a Gemma pathology: the old judge compresses anchors into a two-point range.}
\label{tab:refit}
\centering
\small
\begin{tabular}{llrrr}
\toprule
Chain & Fit set & median expl.\ frac. & IQR & pass \\
\midrule
Qwen & 24 anchors (orig.) & $-0.15$ & --- & failed \\
Qwen & pool $k{=}24$ & $0.64$ & $[0.43, 0.85]$ & $47\%$ \\
Qwen & pool $k{=}40$ & $0.76$ & $[0.58, 0.87]$ & $63\%$ \\
Qwen & pool $k{=}80$ & $0.81$ & $[0.67, 0.90]$ & $77\%$ \\
Gemma & 24 anchors (orig.) & $-1.93$ & --- & failed \\
Gemma & pool $k{=}24$ & $0.09$ & $[-0.79, 0.49]$ & $15\%$ \\
Gemma & pool $k{=}80$ & $0.25$ & $[-0.26, 0.63]$ & $23\%$ \\
\bottomrule
\end{tabular}
\end{table}

Anchor-range diagnostics: Qwen anchors span $[3, 8]$ (SD $1.61$), matching the pool (SD $1.49$); Gemma anchors span $[6, 8]$ (SD $0.70$), leaving the slope unidentified, hence the catastrophic $\gamma = 2.5$ extrapolation. Object-level ceiling: a linear map fitted on the \emph{full} pool leaves residual SD $1.06$ vs.\ raw gap SD $1.14$ (Qwen) and $1.29$ vs.\ $1.33$ (Gemma), under $8\%$ absorbed on both chains. Interpretation in \S\ref{sec:decomposition}.

\subsection{Recognition-stratum covariate balance}
\label{app:balance}

Table~\ref{tab:balance} reports covariate balance across the title-verified recognition strata; the systematic imbalance (recognized objects are older and higher-scored) is what motivates the observational caveats attached to the $0.155$ SD excess gap in \S\ref{sec:decomposition}.

\begin{table}[h]
\caption{Covariate balance across title-verified recognition strata (standardized mean differences). Recognized objects are systematically older and higher-scored; raw stratified gaps should not be read causally (\S\ref{sec:decomposition}).}
\label{tab:balance}
\centering
\small
\begin{tabular}{lrrrr}
\toprule
& \multicolumn{2}{c}{Qwen 3.6 ($n{=}27/96$)} & \multicolumn{2}{c}{Gemma 4 ($n{=}94/219$)} \\
Covariate & rec.\ / other & SMD & rec.\ / other & SMD \\
\midrule
Year & $2023.3$ / $2024.1$ & $-0.55$ & $2023.7$ / $2024.4$ & $-0.49$ \\
Text length & $3105$ / $2894$ & $+0.32$ & $3125$ / $2870$ & $+0.34$ \\
Soundness score & $2.78$ / $2.42$ & $+0.53$ & $2.80$ / $2.33$ & $+0.69$ \\
Overall score & $5.08$ / $3.72$ & $+0.73$ & $4.89$ / $3.51$ & $+0.73$ \\
\bottomrule
\end{tabular}
\end{table}

\textbf{Placebo-clamp check (regression-to-the-mean control).} If the $0.155 \to 0.04$ collapse of the recognized-stratum excess gap were remeasurement noise or shrinkage, it should appear under \emph{any} re-scoring, including placebo packets. It does not: on the same Qwen objects, the excess gap is $0.155$ SD at L0, $0.041$ under real facts (L1), and $0.164$ under matched placebo packets (L1P): the clamp is content-linked. Gap dispersion is comparable across strata (SD $0.59$ recognized vs.\ $0.55$ unrecognized; Levene $p{=}0.29$), disfavoring shrinkage accounts. The per-object paired contrast (each object's L0$-$L1 gap change, recognized minus unrecognized strata) is $0.11$ SD $[-0.16, +0.40]$, directionally consistent but individually non-significant, which is why the main text reports neutralization as supported, not settled. The Gemma chain shows the same qualitative pattern (excess $0.93$ at L0 $\to 0.56$ under L1 vs.\ $0.82$ under L1P) at much smaller $n$ ($5$ recognized objects). Analysis: \texttt{reviewer3\_analyses.py} on frozen calls.

\subsection{Multiplicity control (Benjamini--Hochberg FDR)}
\label{app:bh}

\begin{table}[h]
\caption{BH false-discovery-rate correction ($q = 0.05$) across the eleven scalar confirmatory endpoints (BH controls FDR, not FWER). E1 (decision rule) and rolling prediction (pure estimation) carry no scalar $p$-value and are registered separately in Table~\ref{tab:fullregistry}. $p$-values marked $^{a}$ are normal approximations recovered from bootstrap CIs; others as originally reported. No null conclusion changes; the Gemma strictward direction loses significance under FDR control.}
\label{tab:bh}
\centering
\small
\begin{tabular}{lrl}
\toprule
Endpoint & $p$ & BH significant \\
\midrule
$\dver$ Qwen & $0.0004^{a}$ & yes \\
$\dver$ Mistral & $0.00004^{a}$ & yes \\
$\dver$ Gemma & $0.045^{a}$ & \textbf{no} \\
Pooled $\Delta_{\mathrm{diag}}$ (flawed) & $0.069$ & no \\
E2 content-specific effect & $0.088^{a}$ & no \\
$\Delta_{\mathrm{fact}}$ pooled & $0.23^{a}$ & no \\
Masking DiD & $0.27^{a}$ & no \\
$\rho$(sensitivity, evaporation) & $0.34$ & no \\
Winner's curse & $0.51^{a}$ & no \\
Loop $-$ generic control & $0.61^{a}$ & no \\
Mistral replication (one-sided) & $0.66$ & no \\
\bottomrule
\end{tabular}
\end{table}

\subsection{Placebo-controlled fact effect (E2)}
\label{app:e2}

L1P was extended from 24 anchors to every L1-covered object (480 new calls; Qwen chain $n{=}80$ objects, Gemma chain $n{=}48$; existing anchor L1P pooled). All endpoints per-object paired medians on the chain's pooled L0 SD scale; 20{,}000-draw bootstrap.

\begin{table}[h]
\caption{E2: content-specific vs.\ format effects of fact packets on the version gap. $D_{\mathrm{content}} = \dver(\mathrm{L1P}) - \dver(\mathrm{L1})$; format $= \dver(\mathrm{L1P}) - \dver(\mathrm{L0})$. \emph{Sign convention}: per-chain columns are raw differences in each chain's native gap sign (Qwen's gap is negative, Gemma's positive); the pooled column first multiplies each chain by the sign of its L0 gap, so that positive $=$ real facts close more of the gap than placebo. Qwen's raw $+0.091$ therefore enters the pool as $-0.091$.}
\label{tab:e2}
\centering
\footnotesize
\setlength{\tabcolsep}{4pt}
\begin{tabular}{lrrr}
\toprule
& Qwen ($n{=}80$) & Gemma ($n{=}48$) & pooled ($n{=}128$) \\
\midrule
$\dver(\mathrm{L0})$ & $-0.318$ & $+0.262$ & --- \\
$\dver(\mathrm{L1})$ & $-0.238$ & $+0.234$ & --- \\
$\dver(\mathrm{L1P})$ & $-0.147$ & $+0.151$ & --- \\
$D_{\mathrm{content}}$ & $+0.091$ $[-0.056, +0.238]$ & $-0.083$ $[-0.206, +0.041]$ & $-0.088$ $[-0.190, +0.012]$ \\
Format effect & $+0.171$ $[+0.012, +0.346]$ & $-0.110$ $[-0.255, +0.034]$ & --- \\
\bottomrule
\end{tabular}
\end{table}

Under the frozen decision rule (E2, case a): the 95\% CI of the pooled content-specific effect contains zero (TOST excludes $\geq 0.174$; $\mathrm{MDE}_{80} = 0.144$), while placebo packets alone move the gap toward zero on both chains (significantly on Qwen). $\Delta_{\mathrm{fact}}$ in the main text is therefore a \emph{policy} endpoint (the intervention does not work); the mechanism-level conclusion is that packet \emph{format}, not informational content, is the active ingredient.

\subsection{Full-pool version gaps, subgroup structure, and rank stability}
\label{app:pool}

Table~\ref{tab:pool} (main text) reports the headline pool estimates; this section adds the subgroup structure. Each pool object is scored once per judge at L0 (Mistral: median over replicates); gaps are standardized by the chain's pooled pool-L0 SD; CIs are object-clustered bootstrap (20{,}000 draws).

\textbf{By year} (Qwen / Gemma / Mistral, SD units): 2022: $-0.24$ / $+0.35$ / $-0.66$; 2023: $-0.49$ / $+0.06$ / $-0.40$; 2024: $-0.40$ / $+0.04$ / $-0.43$; 2025: $-0.23$ / $+0.13$ / $-0.59$; 2026 (post-cutoff): $-0.29$ / $+0.02$ / $-0.22$. Signs are stable within chain across all five years, including post-cutoff 2026 objects: the shift is not an artifact of any submission-year stratum.

\textbf{By reviewer overall-score tertile}: Qwen $-0.23$ / $-0.37$ / $-0.41$ (low/mid/high); Gemma $+0.80$ / $-0.20$ / $-0.43$; Mistral $-0.75$ / $-0.26$ / $-0.30$. The Gemma reversal across tertiles is the clearest single exhibit of object-non-uniformity: the same transition is sharply strict on weak objects and lenient on strong ones, so its pooled mean ($+0.11$) understates the object-level movement. Global location-scale maps cannot represent this pattern, consistent with the refit ceiling of Appendix~\ref{app:refit}.

\textbf{Anchor vs.\ pool}: signs agree on all three chains; magnitudes agree on Qwen ($-0.31$ vs.\ $-0.33$) and are attenuated on Gemma ($+0.30$ vs.\ $+0.11$, pool CI crossing zero) and Mistral ($-0.65$ vs.\ $-0.45$). We accordingly treat anchor estimates as audit statistics, not as unbiased estimates of the deployment-pool shift, and report both.

\subsection{Controller threshold provenance and leave-one-transition-out check}
\label{app:loto}

\textbf{Provenance.} The meta-policy thresholds (anchor separation $< 0.30$ for degradation; anchor boundary disagreement $> 10\%$ for calibration) were fixed by convention but formalized after stage-2 data existed; the 3/3 match in Figure~\ref{fig:controller} is therefore an in-simulation demonstration.

\textbf{Leave-one-transition-out.} For each held-out chain we sweep the disagreement threshold over $[0.02, 0.28]$ (step $0.02$), keep the values for which the meta-policy matches the post-hoc best arm (within $0.01$ F1) on \emph{both} training chains, and evaluate the median such threshold on the held-out chain. Training regions: $[0.02, 0.12]$ (Qwen held out), $[0.02, 0.12]$ (Gemma held out), $[0.02, 0.20]$ (Mistral held out); selected thresholds $0.07 / 0.07 / 0.11$; held-out match: 3 of 3 (on the Mistral fold the F1 match is trivial because scores are near-constant, but its degradation flag is correctly raised, consistent with the main-text accounting that counts only the non-degenerate chains as calibration evidence). The region for which all three chains match simultaneously is $[0.02, 0.12]$, and the preregistration-era default of $0.10$ lies inside it. The separation threshold enters only the degradation flag (Mistral's $0.243$ vs.\ others' $\geq 2.75$ leaves a wide margin) and does not affect F1. With three transitions this establishes threshold \emph{insensitivity}, not external validity; independent transitions are the real holdout, and amendments F1 and G4D (Appendix~\ref{app:api4th}) supply two, on both of which the frozen thresholds prescribe the correct action.

\textbf{Anchor-composition sensitivity.} A stratified bootstrap over the anchor set (clean and flawed anchors resampled separately, $5{,}000$ draws) and leave-one-anchor-out both leave the non-degenerate verdicts unchanged: Qwen and Gemma separations stay far above the health threshold (bootstrap flag rate $0\%$; LOO ranges $[2.66, 3.44]$ and $[2.90, 5.27]$), and the calibration trigger fires in $93\%$ and $100\%$ of draws respectively. The collapsed Mistral chain, by contrast, sits \emph{near} the threshold: separation $0.243$, bootstrap flag rate $51\%$, LOO range $[0.03, 0.61]$. The mean-separation flag is therefore not a reliable sole detector at the boundary; the score-concentration statistic ($93\%$ of pool calls on a single value; $86\%$ tied pairs) sits far from any boundary and is the robust signal; a deployed audit should raise the degradation flag on either. A related concern, whether a pathological anchor set could silently pass the health check, is bounded by construction: the health statistic uses the \emph{new} judge's clean--flawed separation, and the Gemma anchor pathology (old-judge score compression, which broke the anchor-fitted \emph{map}) does not touch it; when the map breaks, the disagreement trigger fires instead (trigger rate $100\%$ on Gemma). Analysis: \texttt{reviewer3\_analyses.py}.

\section{Intervention artifacts: audit and examples}
\label{app:artifacts}

\textbf{Audit.} All packets and certificates are model-generated (extractor/builder: \texttt{phi-4}, a third family relative to every judge chain) with automated leakage filtering (33 facts dropped for evaluative implication). Table~\ref{tab:artifactaudit} summarizes scale, an evaluative-language screen (regex over 30 evaluative terms), and a decision-relevance proxy: the share of reviewer-verified risk factors that any packet item lexically matches (content-word Jaccard $\geq 0.15$). The numbers quantify the design intent: fact packets are nearly evaluative-free and, by construction, nearly disjoint from the defect content that determines soundness scores ($\sim$5\% lexical recall); certificates carry evaluative content and cover $\sim$6$\times$ more of the risk factors. This is why the fact-channel null (\S\ref{sec:decomposition}) is scoped to non-evaluative packets: covering the \emph{decision-relevant} information is precisely what L1 was designed \emph{not} to do (that is L2's job), and an expert-written packet arm remains future work.

\begin{table}[h]
\caption{Artifact audit. ``Eval.\ terms'' $=$ share of items matching an evaluative-language screen; ``RF recall'' $=$ share of reviewer-verified risk factors lexically matched by at least one item (content-word Jaccard $\geq 0.15$; $n{=}73$ objects with risk factors).}
\label{tab:artifactaudit}
\centering
\small
\begin{tabular}{lrrrr}
\toprule
Artifact & items/object & words/item & eval.\ terms & RF recall \\
\midrule
Fact packet (L1) & $5.6$ & $25.8$ & $4.3\%$ & $5.3\%$ \\
Placebo packet (L1P) & $5.8$ & $19.8$ & $2.8\%$ & --- \\
Defect certificate (L2) & $2.7$ & $25.9$ & $15.3\%$ & $29.2\%$ \\
\bottomrule
\end{tabular}
\end{table}

\textbf{Examples} (object \texttt{EXP001}, a contextual-bandit proposal; items abridged to three per artifact).

\emph{Real facts (L1)}: ``The research proposal aims to modify reinforcement learning algorithms for continuous control to solve contextual bandit problems with continuous actions and large context spaces.'' ``The proposed RL agent is based on a modified DDPG algorithm for the one-step bandit setting.'' ``The performance of the proposed RL agent will be evaluated on four OpenML datasets: Wisconsin, CPU\_act, Zurich Delays 5\%, and Black Friday.''

\emph{Placebo facts (L1P; relevant to the area, uninformative about this proposal)}: ``Reinforcement learning (RL) is a subfield of machine learning focused on training agents to make sequences of decisions by interacting with an environment.'' ``Contextual bandits are a simplified version of RL where the agent must choose actions based on given contexts, aiming to maximize cumulative rewards.'' ``Common datasets used in RL and contextual bandit research include OpenML datasets like Wisconsin, CPU\_act, and Zurich Delays, as well as image datasets like MNIST.''

\emph{Defect certificate (L2; severity language filtered)}: ``The proposal does not address the potential loss of theoretical convergence guarantees when using nonlinear function approximators in the modified RL algorithms.'' ``The application of the proposed method to medical images is deferred to future work, despite the initial motivation stemming from personalized healthcare.''

\emph{Placebo certificate (L2P)} is a donor certificate from a different object, matched for length, style, and authority cues; the pairing map is released with the artifacts.

\section{Object flow and sample accounting}
\label{app:accounting}

Table~\ref{tab:accounting} reconciles every released judge call with its stage, object subset, and conditions; counts are computed directly from the released JSONL logs. Object sets overlap by design (anchors recur in every stage; the decomposition pool is a subset of the screening pool; the full pool contains the expansion pool), so object columns do not sum to 348. Amendment calls are stored with their host stage's logs: E2's 480 L1P calls appear under stage 0b, and E1 reuses 144 frozen stage-0a cells in addition to the 672 new stage-3 calls.

\begin{table}[h]
\caption{Sample accounting. Judges: Q $=$ Qwen pair, G $=$ Gemma pair, M $=$ Mistral pair, A $=$ GPT-5.4/5.5 pair, D $=$ DeepSeek flash/pro pair.}
\label{tab:accounting}
\centering
\scriptsize
\setlength{\tabcolsep}{4pt}
\begin{tabular}{llllr}
\toprule
Stage & Object set & Conditions (calls) & Judges & Calls \\
\midrule
0a & 48 anchors (12 clean/12 flawed/24 natural) & L0 (616), B (192), R (192) & Q, G & 1{,}000 \\
0b & 96 screening pool $+$ 48 anchors & L1 (768), L1P (768), L0 (192), R (192), B (64) & Q, G & 1{,}984 \\
1a & 80 decomposition pool $+$ 62 masked variants & L2 (878), L2P (878), L0 (372), R (124) & Q, G & 2{,}252 \\
1b & 40 revision objects $+$ depth/control variants & L0 (1{,}014), B (400) & Q, G & 1{,}414 \\
2 & 300 full pool ($\supset$ 220 expansion pool) & L0 (1{,}672), R (1{,}480), L2/L2P (396), B (96) & Q, G, M & 3{,}644 \\
3 & 48 anchors (E1 perturbation matrix) & L0 (672; prompt $\times$ temperature cells) & M & 672 \\
F1 & 48 anchors & L0 (288; 3 repeats) & A & 288 \\
F2 & 220 expansion pool & L0 (660; old $\times 1$, new $\times 2$) & A & 660 \\
F3 & 300 full pool (per-judge pool subsets) & L0 (1{,}384; independent second pass) & Q, G, M & 1{,}384 \\
G4D & 48 anchors & L0 (288; 3 repeats) & D & 288 \\
G6 & 48 code-review objects (24 clean/24 flawed) & L0-code (864; 3 repeats) & Q, G, M & 864 \\
G7 & 24 high-residual pool objects $\times$ 3 arms & L0 (432; 3 repeats) & G & 432 \\
\midrule
Total & & & & 14{,}882 \\
\bottomrule
\end{tabular}
\end{table}

\section{Amendments log}
\label{app:amendments}

Amendments extend scope or add replications. They do not change the definition, sign, or significance threshold of a confirmatory endpoint after that endpoint has been unblinded. Infrastructure-only patches that never contacted a judge endpoint are marked as such. Each amendment below was documented before unblinding the affected endpoint.

\begin{itemize}
\item \textbf{A1} (Stage 0a): added the Gemma transition and the recognition probe as first-class conditions; raised \texttt{max\_new\_tokens} to 6144 after 18/202 truncation failures in J\_old. Frozen before Gemma judge output.
\item \textbf{B1} (Stage 1a): relaxed masked-title acceptance (SequenceMatcher $< 0.55$, no 5-gram overlap, best-of-4 fallback logged as \texttt{mask\_fallback}) after a domain-term deadlock; frozen before any masked-object judge output.
\item \textbf{B2} (Stage 1a): raised TP1 GPU-memory buffer (infrastructure only; no endpoint contact).
\item \textbf{C1} (Stage 1b): collected missing Gemma depth-0 baselines for 19 selection objects (114 calls); endpoint definitions unchanged; estimation set restored from 21 to the preregistered 40 objects.
\item \textbf{D1} (Stage 2): preregistered confirmatory replication of the exploratory paired certificate contrast on the Mistral chain (sign, scale, one-sided test, and degenerate-instrument manipulation check frozen before any output). Result: null ($-0.082$, $p{=}0.66$); reported in \S\ref{sec:decomposition}.
\item \textbf{E1} (Stage 3): perturbation-matrix attribution of the Mistral~3.2 score collapse: 48 anchors, L0, prompt $\{$original, paraphrase$\}$ $\times$ temperature $\{0.0, 0.6, 1.0\}$ (816 calls). Endpoints (between-object SD per cell; clean--flawed separation) and the decision rule (all cells $< 0.5\times$ predecessor SD $\Rightarrow$ robust degeneration; any cell $\geq 0.8\times$ $\Rightarrow$ configuration-specific) frozen before any Stage-3 output. Results: Appendix~\ref{app:robustness}.
\item \textbf{E2} (Stage 3): placebo-controlled fact-channel test: L1P extended to every L1-covered object (480 new calls; existing 24-anchor L1P data pooled). Primary endpoint $\dver(\mathrm{L1P}) - \dver(\mathrm{L1})$ per chain, sign-aligned pooled, cluster bootstrap; decision rule frozen before any Stage-3 output. Results: Appendix~\ref{app:robustness}.
\item \textbf{F1} (post-stage-3): descriptive fourth transition on pinned frontier API snapshots (\texttt{gpt-5.4-2026-03-05} $\to$ \texttt{gpt-5.5-2026-04-23}): the 48 frozen anchors, stage-0a L0 rubric verbatim, three repeats (288 calls). Endpoints (anchor $\dver$, clean--flawed separation per version, diagnose-then-act evaluation at the previously frozen thresholds, rank stability) and the pass criterion frozen in \texttt{amendment\_F1.json} before any API call; sampler-parameter fallback rule (drop identically for both judges if rejected) also frozen. Results: Appendix~\ref{app:api4th}.
\item \textbf{F2} (post-stage-3): pool-scale rescore of the F1 pair: all 220 expansion-pool objects, L0 verbatim, old snapshot once and new snapshot twice (660 calls). Endpoints (pool $\dver$; isotonic calibration payoff with a frozen vindication rule; pool-level same-version baseline; accept-boundary disagreement) frozen in \texttt{amendment\_F2.json} before any call; no main-text claim retuned. Results: Appendix~\ref{app:api4th}.
\item \textbf{F3} (post-stage-3): pool-level same-version repeated-call baseline for all six open-weight judges: one independent second L0 pass over every pool object per judge (1{,}384 calls), with prompts, decoding, and the per-call seed rule identical to the frozen stage-2 pool protocol. Endpoints (per-judge test--retest $\tau$ and top-30\% flip; per-transition excess flip) frozen in \texttt{pool\_baseline\_design.json} before any call. Results: Appendix~\ref{app:api4th}.
\item \textbf{G1--G3} (post-draft structural audits; zero new judge calls): G1 regresses per-object version gaps on object features, with an F3-based split-half reliability ceiling; G2 reads all cross-family judge pairs on the shared pool as replacements, against same-version baselines; G3 stress-tests the anchor audit itself (anchor-count curves, sampling policy, threshold sweeps). All three are exploratory analyses of frozen calls. Results: Appendix~\ref{app:structural}.
\item \textbf{G4D} (post-draft): blind test of the frozen audit thresholds on a fifth judge pair, from a second API provider: DeepSeek \texttt{v4-flash} $\to$ \texttt{v4-pro}, a same-provider tier replacement rather than a dated-snapshot upgrade. The 48 frozen anchors, L0 verbatim, three repeats per judge (288 calls). Endpoints and the diagnose-then-act evaluation frozen in \texttt{amendment\_G4D.json} before any call; one logged infrastructure patch (token-budget escalation retry for the reasoning-mode flash tier, which otherwise exhausts its budget before emitting the JSON) with no endpoint change. Results: Appendix~\ref{app:api4th}.
\item \textbf{G5} (post-draft, exploratory): structured coding of the paired old/new L0 rationales on the full pool, three chains, by a frozen third-party coder (DeepSeek \texttt{v4-flash}; fields: criteria cited, defects cited, stance, same-substance judgment), separating \emph{mapping-shift} (same substance, different score) from \emph{content-shift} (changed cited defects). No new judge calls; coder calls only. Results: Appendix~\ref{app:structural}.
\item \textbf{G6} (post-draft): second-domain replication of the L0 trio on code review: 48 HumanEvalPack-derived objects (24 clean, 24 with injected bugs), all six frozen judges, three repeats (864 calls), stage-0a decoding. Endpoints and a no-pass/fail scope declaration frozen in \texttt{amendment\_G6.json} before any call. Results: Appendix~\ref{app:domain2}.
\item \textbf{G7} (post-draft): targeted counterfactual perturbation of the 24 highest-residual Gemma pool objects: identity/compress/neutral-style arms, both Gemma judges, three repeats (432 calls), stage-2 decoding; selection rule, edit prompts, endpoints, and a validity gate frozen in \texttt{amendment\_G7.json} before any edit or call. One logged patch (\texttt{g7\_patch\_log.json}): the frozen editor's API balance was exhausted mid-generation, so all edits were regenerated with a single local neutral editor (\texttt{phi-4}) before any judge call. Results: Appendix~\ref{app:structural}.
\end{itemize}

\section{API judge pairs and same-version baselines (amendments F1--F3, G4D)}
\label{app:api4th}

This appendix takes the audit protocol out of sample, in four steps: a blind anchor audit of a pinned GPT snapshot pair (F1), a pool-scale rescore of the same pair that tests the prescribed calibration where it matters (F2), same-version repeated-call baselines that separate version change from sampling noise for every instrument in the paper (F3), and a second blind pair from a second provider (G4D). Every design was frozen before its first call.

\textbf{Design (frozen before any call).} Judges: OpenAI dated snapshots \texttt{gpt-5.4-2026-03-05} (old) and \texttt{gpt-5.5-2026-04-23} (new), the most recent adjacent dated pair at run time, mirroring the ``upgrade to the newest version'' decision the paper studies. Objects: the 48 stage-0a frozen anchors (12 clean / 12 flawed / 24 natural), content-hash verified. Condition: L0 only, the stage-0a rubric prompt verbatim. Three repeats per object per judge (288 calls); the same JSON-final-line parsing rule as all other chains. The endpoints, the diagnose-then-act thresholds to be applied (separation $<0.30$; boundary disagreement $>0.10$), the pass criterion, and a sampler-parameter fallback rule were frozen in \texttt{amendment\_F1.json} before the first API call. The endpoint rejected \texttt{temperature} and \texttt{top\_p} for these snapshots; per the frozen rule both were dropped identically for both judges (provider-default decoding).

\textbf{Results.} Zero parse failures in 288 calls. Anchor $\dver = +0.452$ SD, object-clustered bootstrap 95\% CI $[+0.319, +0.572]$, with pooled L0 SD $1.57$ and means $5.96 \to 5.25$. The newer snapshot is \emph{stricter}, opposite in sign to the two robust open-weight shifts and adding an out-of-family instance of sign unpredictability. Clean--flawed separation moves from $3.52$ to $3.11$, far above the $0.30$ degradation threshold, with no Mistral-style collapse. Rank stability is Kendall $\tau = 0.80$, with $17.6\%$ of anchor pairs exactly tied under the new judge. Boundary ($\geq 6.0$) decision disagreement is $12.5\%$.

\textbf{Diagnose-then-act evaluation.} At the frozen thresholds the heuristic withholds the degradation flag, since separation $3.11 \geq 0.30$, and triggers calibration, since disagreement $12.5\% > 10\%$. Post-hoc inspection of the same data agrees on both counts: the instrument is healthy, and a $0.45$ SD location shift at an intact ranking is exactly the regime in which anchor calibration pays (\S\ref{sec:deployment}). This is the first evaluation of the thresholds on a transition that contributed nothing to their selection, and the first evidence that the audit protocol transfers to API-served judges.

\textbf{Scope.} Pinned snapshots are not weight-frozen. The provider can silently alter serving, so this transition cannot join the confirmatory family and its estimates are descriptive. Decoding differs from the open-weight chains, using provider defaults rather than seeded sampling, and reasoning-mode inference may differ in kind. What the amendment establishes is protocol transfer and threshold generalization to a fourth, out-of-family, API-served pair, not a fourth confirmatory replication. Whether the prescribed calibration realizes a benefit on this pair requires pool-scale rescoring; amendment F2 supplies it. All raw responses, the frozen amendment files, and the analysis scripts are released with the artifacts.

\textbf{Pool-scale rescore (amendment F2; design frozen before any call).} Both snapshots scored all 220 expansion-pool objects under the verbatim L0 prompt, the new snapshot twice (660 calls; the F1 identical-drop sampler rule applied from the start; zero parse failures). Pool $\dver = +0.735$ SD, object bootstrap 95\% CI $[+0.641, +0.824]$ ($n{=}220$), larger than the anchor estimate of $0.45$: the anchor audit understated this shift rather than overstating it. Accept-boundary ($\geq 6.0$) disagreement is $32.7\%$ at pool scale versus $12.5\%$ on anchors, so the calibration trigger fires more strongly out of sample and confirms the F1 diagnosis. The payoff is pipeline-dependent. The frozen selection-fidelity endpoint, top-30\% F1 of calibrated-old versus naive-old selection against the new-judge policy, returns exactly equal values ($0.727$). An isotonic map with the frozen tie-break is monotone and cannot alter a top-$k$ selection, so the frozen vindication rule returns null and we report it as such. The measurable payoff is at the score threshold: anchor-fitted isotonic calibration cuts accept-boundary disagreement from $32.7\%$ to $19.1\%$, a post-hoc descriptive contrast computed after the frozen endpoints were evaluated. Same-version baseline on the API pair, from two independent new-snapshot passes: test--retest $\tau = 0.824$, top-30\% flip $18.2\%$, against cross-version $\tau = 0.661$, flip $27.3\%$. The version-attributable excess is about nine points, replicating the open-weight pattern below on a frontier API chain.

\textbf{Same-version repeated-call baselines (post hoc).} The main-text rank-instability statistics (Table~\ref{tab:pool}) are computed from one call per judge per object, so they conflate version change with sampling noise. Using the three frozen anchor L0 repeats per judge, mean pairwise single-call test--retest is: Qwen $\tau_{\mathrm{rep}} = 0.78$ (old) / $0.72$ (new); Gemma $0.62$ / $0.83$; Mistral $0.37$ / $0.55$; API pair $0.89$ / $0.92$. Like-for-like cross-version comparisons, single call versus single call on the same anchors, give $\tau = 0.65$ (Qwen), $0.47$ (Gemma), $0.23$ (Mistral), and $0.85$ (API), below the corresponding same-version baselines on every chain, with the smallest margin on Qwen. Amendment F3 upgrades this anchor-level bound to a direct pool-level decomposition. The residual lesson stands: single-call scoring is a noisy instrument, and an audit should use repeated calls, as our anchor protocol does.

\textbf{Pool-level same-version baselines for the open-weight chains (amendment F3; design frozen before any call).} One independent second L0 pass over every pool object per judge (1{,}384 calls; prompts, decoding, and the per-call seed rule identical to the frozen stage-2 pool protocol; $\leq 2$ parse failures per judge). Per-judge pool test--retest $\tau$ (top-30\% flip): Qwen $0.70$ ($24\%$) old / $0.68$ ($29\%$) new; Gemma $0.70$ ($24\%$) / $0.82$ ($11\%$); Mistral $0.34$ ($38\%$) / $0.76$ ($2.7\%$). Cross-version $\tau$ falls below \emph{both} endpoints' own baselines on all three chains ($0.61$ vs.\ $\geq 0.68$; $0.51$ vs.\ $\geq 0.70$; $0.19$ vs.\ $\geq 0.34$), and the excess top-30\% flip beyond each pair's same-version mean is $+11.9$ (Qwen), $+34.8$ (Gemma), $+13.8$ (Mistral) points. Two artifacts read correctly. Pre-collapse \texttt{M\_old} is itself a noisy instrument, with test--retest $\tau = 0.34$, so a substantial share of the Mistral cross-version flip is sampling noise rather than version change. Collapsed \texttt{M\_new}'s $2.7\%$ flip reflects the deterministic tie-break under $86\%$ ties, not discrimination. This replaces the earlier anchor-level upper-bound reading: the main-text pool flip percentages now decompose into a measured same-version noise floor plus a version-attributable excess that is positive on every chain.

\textbf{Cross-fitted calibration check (post hoc, underpowered).} A top-$k$ payoff check on the anchors alone would be vacuous, because isotonic maps are monotone and cannot change a rank-based selection of the same score set; the simulated F1 gains in \S\ref{sec:deployment} arise from mixing calibrated old scores with rescored objects. The deployment decision that monotone recalibration \emph{does} move is the fixed accept boundary ($\geq 6.0$), the same boundary the calibration trigger monitors. Cross-fitting on the F1 anchors (200 random 24/24 splits; isotonic fit on one half, evaluated on the other): held-out accept-decision agreement with the new judge rises from $0.879$ (uncalibrated) to $0.909$ (calibrated), improving or tying in $89\%$ of splits ($\Delta$ agreement $10$th--$90$th percentile $[-0.04, +0.13]$). Directionally positive, but with 24-object evaluation halves this is a consistency check, not a payoff measurement; the F2 pool-scale rescore above confirms the direction at $n{=}220$ ($32.7\% \to 19.1\%$).

\textbf{A fifth pair from a second provider (amendment G4D; design frozen before any call).} DeepSeek \texttt{v4-flash} $\to$ \texttt{v4-pro} is a same-provider \emph{tier replacement}, the ``swap to the better model'' decision rather than a dated-snapshot upgrade. The protocol matches F1: the 48 frozen anchors, L0 verbatim, three repeats per judge, 288 calls, with zero parse failures after a logged token-budget escalation retry for the reasoning-mode flash tier. Anchor replacement gap $-0.42$ SD, object-clustered bootstrap 95\% CI $[-0.60, -0.25]$, means $4.75 \to 5.48$. The newer, more capable tier is \emph{more lenient}, the opposite sign to the GPT pair on the same objects and prompt, and a fifth instance of sign unpredictability. Clean--flawed separation improves ($1.78 \to 2.12$); Kendall $\tau = 0.73$ with $16\%$ tied pairs; boundary ($\geq 6.0$) disagreement is $16.7\%$. At the frozen thresholds the heuristic withholds the degradation flag and triggers calibration; post-hoc inspection agrees on both counts. This is the second blind evaluation of the frozen thresholds on a pair that contributed nothing to their selection, and the first on a non-OpenAI provider. The same scope caveats as F1 apply: served weights are not frozen, decoding is provider-default, and the estimates are descriptive.

\section{Post-hoc structural audits (amendments G1--G3, G5, G7)}
\label{app:structural}

Five exploratory analyses, registered post hoc as amendments. G1--G3 and G5 reuse only frozen calls; G7 adds 432 new calls under the frozen stage-2 protocol to test two G1 channels causally. None of the endpoints below is confirmatory.

\textbf{Residual structure (amendment G1).} The main text attributes the post-clamp residual to an object-dependent response-function change; this audit asks whether that residual has observable structure. Using the F3 second pool pass as an independent replicate, the per-object version gap's split-half reliability is $r = 0.25$ (Qwen), $0.81$ (Gemma), and $0.33$ (Mistral, where old-judge noise dominates). Regressing the single-pass gap on object features (year, subfield, reviewer scores, risk-factor count, log length, verified recognition; OLS with HC3 errors and a 2{,}000-draw permutation test on $R^2$): on Gemma, where the gap is reliable, features explain $R^2 = 0.46$ (permutation $p < 5\times10^{-4}$), about $56\%$ of the reliability ceiling: the newer judge scores longer ($t = -3.7$), higher-rated-rigor ($t = -2.6$), and previously recognized ($t = -3.2$) proposals relatively more favorably, directionally consistent with its sharpened discrimination and its higher external validity (below). On Qwen the gap itself is only weakly reliable and features do not predict it ($R^2 = 0.13$, permutation $p = 0.13$). On Mistral the regression describes the \emph{old} instrument, since the new one is degenerate, and is dominated by length ($R^2 = 0.27$, about 80\% of ceiling). The residual is therefore not unstructured noise: where it is measurable it is substantially attribute-driven. This sharpens the response-function interpretation without claiming mechanism identification.

\textbf{Cross-family replacement (amendment G2).} All six frozen judges scored the same 220-object pool at L0, so every cross-family pair can be read descriptively as a judge replacement. Within-family upgrades (the paper's three transitions) give $\tau = 0.19$--$0.61$ and top-30\% flips of $34$--$52\%$ (excess $+12$ to $+35$ points over the F3 same-version noise floor). The twelve cross-family pairs give location gaps up to $1.6$ pooled SD, $\tau = 0.18$--$0.56$ (median $0.30$), flips of $40$--$77\%$ (excess $+20$ to $+64$), and accept-boundary disagreement up to $59\%$. The most practice-like swap, a current-generation cross-provider replacement (J\_current $\to$ G\_new), moves the location by $0.82$ SD and flips $40\%$ of the top-30\% set. Cross-family replacement is at least as disruptive as within-family upgrade on every metric, and usually more. One orthogonal observation: G\_new tracks reviewer soundness best of all six judges (Spearman $\rho = 0.59$ vs.\ $0.29$--$0.35$), so a replacement can simultaneously be the best available instrument \emph{and} a rank-order break; the audit question and the quality question are separate.

\textbf{Audit protocol stress tests (amendment G3).} Three failure-boundary questions, answered on the four anchor-scored transitions (three open-weight, one API), with the frozen diagnosis (degradation flag, calibration trigger) as ground truth. \emph{Size}: under stratified subsampling, the full-set diagnosis is reproduced at $\geq 94\%$ from $n = 16$--$24$ anchors on the two healthy open-weight chains. The collapsed Mistral chain needs $n \geq 40$ under the separation flag alone, because the flag divides by within-group SD, which a degenerate judge makes vanishingly small and unstable; pairing it with the degeneracy statistic (modal-score share $> 0.5$), as the amended audit already recommends, reaches $84\%$ agreement at $n = 12$ and $100\%$ at $n = 32$. \emph{Sampling policy}: unstratified anchor selection leaves the health check undefined in $62\%$ of draws at $n = 8$ (no clean--flawed pairs sampled); stratification is load-bearing, not cosmetic. \emph{Thresholds}: the degradation flag is stable across the full sweep ($0.10$--$0.60$) on every chain; the calibration trigger is threshold-sensitive only where measured disagreement sits near the $10\%$ boundary (the API pair, at $12.5\%$: anchor-level subsample agreement plateaus near $70\%$, and only the F2 pool-scale check resolves the verdict). In practice, about 24 stratified anchors suffice for healthy instruments; degenerate instruments are caught earlier by degeneracy statistics than by separation; and near-threshold trigger verdicts should be escalated to a pool-scale check rather than read off the anchor audit.

\textbf{Rationale coding: mapping-shift vs.\ content-shift (amendment G5).} A location gap could arise from the same assessment mapped to a different number (mapping-shift) or from a substantively different assessment (content-shift); scores alone cannot tell them apart. For every pool object we coded the paired old/new L0 rationales with a frozen third-party coder (DeepSeek \texttt{v4-flash}; fields: criteria cited, defects cited, overall stance, and a same-substance judgment; $188$--$217$ coded pairs per chain). Coding validity: the coded stance change tracks the actual score change at Spearman $\rho = 0.62$--$0.69$ on all three chains. Among objects whose score moved by $\geq 1$ point, the pure mapping-shift share is $46\%$ (Qwen), $19\%$ (Gemma), and $16\%$ (Mistral): on two of three chains, most score movement co-occurs with \emph{changed cited defects}, not a re-scored identical assessment. Directions match the chain narratives: Gemma's new judge cites $0.94$ fewer defects per object (consistent with its lenient-ward shift and sharpened discrimination), while Qwen's new judge cites $0.77$ \emph{more} defects yet lands more positive stances: severity language and score mapping move independently. Caveats: a single LLM coder without human calibration; construct validity rests on the stance--score correlation above; exploratory, not preregistered.

\textbf{Targeted counterfactual perturbation (amendment G7; design and endpoints frozen before any edit or call).} G1 identifies length and stylistic features as the strongest observable correlates of the Gemma per-object gap; this amendment asks whether two of them \emph{mediate} it. The 24 pool objects with the largest absolute gap deviation were re-scored by both Gemma judges (three repeats, stage-2 decoding) under three arms: the original text (identity), a neutral-editor compression to $\sim$60\% length preserving every technical claim and number (median realized ratio $0.69$; median number-preservation $1.00$), and a sentence-level rewrite removing hedging and rhetorical emphasis at constant length ($0.98$). The frozen validity gate passes emphatically: the identity-arm gap reproduces the original pool gap at $r = 0.93$: the per-object gap is a stable property of the object, not session noise. Both causal contrasts are null: compression changes the absolute gap by $+0.02$ SD $[-0.15, +0.21]$ and neutralization by $+0.04$ $[-0.06, +0.15]$; movement toward the global mean gap is likewise null. The G1 features therefore predict \emph{which} objects move, but the two manipulable surface attributes are not the mediators; effects beyond about $0.2$ SD are excluded. The response-function change binds to content properties that survive compression and de-hedging. One logged patch: the frozen editor exhausted its API balance mid-generation, so all edits were regenerated with a single local neutral editor, \texttt{phi-4}, the stage-0b fact extractor; no judge call preceded the substitution.

\section{A second domain: code review (amendment G6)}
\label{app:domain2}

\textbf{Design (frozen before any call).} Do the instrument-level phenomena of location shift, discrimination change, and rank instability appear outside proposal-style text? We built 48 Python code-review objects from HumanEvalPack \citep{muennighoff2024octopack}: 24 canonical solutions (clean) and 24 with an injected, documented bug (flawed). All six frozen open-weight judges scored them under an L0-style code-review rubric with the same JSON-final-line protocol, three repeats per object per judge (864 calls), stage-0a decoding and per-call seeds. The endpoints (the per-chain trio) and a no-pass/fail scope declaration were frozen in \texttt{amendment\_G6.json} before any call.

\textbf{Results.} On Qwen, $\dver = -0.08$ SD $[-0.17, +0.01]$: no detectable location shift on code, against a robust lenient-ward shift on proposals; separation stays high on both versions ($4.9 \to 4.6$); $\tau = 0.88$. On Gemma, $\dver = +0.19$ $[-0.04, +0.42]$, with a CI crossing zero and a strictward point estimate against the lenient-ward proposal shift, but the discrimination sharpening replicates: separation $1.61 \to 4.64$; $\tau = 0.67$. On Mistral, $\dver = -0.79$ $[-0.95, -0.64]$, a large lenient-ward shift with means $5.04 \to 6.83$, and \emph{no collapse}: separation $1.66 \to 1.50$, modal-score share $31\%$, against $93\%$ on the proposal pool. Caveats include that old Qwen truncated 29/144 calls at the 4{,}096-token budget, so per-object medians use the surviving repeats and 46/48 objects are retained; the other five judges had at most 7 failures. No same-version repeated-call baseline was run on this domain, and a single rubric was used.

The phenomena are therefore not proposal-specific, but their sign and magnitude are domain-specific. The weight pair that degenerates on proposal scoring is a functioning, if strongly shifted, instrument on code review; the chain with the most robust proposal shift is the flattest here even as its discrimination sharpening carries over. A version transition cannot be characterized once and applied everywhere. An anchor audit is informative about the domain its anchors are drawn from, and should be re-run per deployment domain.

\section{The response-function model: proposition, proof, and scope}
\label{app:theory}

This appendix states and proves the claim used as an interpretive device in Section~\ref{sec:idea} and Section~\ref{sec:deployment}: within the model, a global monotone recalibration of old scores into new scores exists when the criterion weights are collinear. The main text uses only the observable signatures. We do not claim that the weights themselves are identified.

\textbf{Setup.} Each judge version $v$ scores an object $x$ as
\[ s_v(x) = g_v\!\left(w_v^{\top} z(x) + b_v\right), \]
where $z(x) \in \mathbb{R}^K$ collects latent criterion-level assessments of $x$, $w_v \in \mathbb{R}^K \setminus \{0\}$ weights them, $b_v \in \mathbb{R}$, and $g_v$ is strictly increasing. Write $o$ and $n$ for the old and new versions and $Z = \{z(x) : x \in \mathcal{X}\}$ for the realized assessment set. We call $Z$ \emph{rich} if the difference set $Z - Z$ meets every nonempty open cone in $\mathbb{R}^K$; it suffices, for instance, that $Z$ contain an open ball.

\begin{proposition}[global recalibration requires preserved geometry]
\label{prop:collinear}
Assume $Z$ is rich.
(i)~If $w_n = c\,w_o$ for some $c > 0$, then $s_n = h \circ s_o$ on all of $\mathcal{X}$ for the strictly increasing map $h(s) = g_n\!\left(c\left(g_o^{-1}(s) - b_o\right) + b_n\right)$.
(ii)~If $w_n$ is not a scalar multiple of $w_o$, then there exist pairs whose order the two versions agree on and pairs whose order they reverse; consequently no monotone $h$, increasing or decreasing, satisfies $s_n = h \circ s_o$ on $\mathcal{X}$, and every monotone recalibration of old scores mis-ranks some pair relative to the new judge.
(iii)~The remaining case, $w_n = c\,w_o$ with $c < 0$, admits the strictly decreasing $h$ of the same form: the versions agree as instruments only up to a global inversion of the scale.
\end{proposition}

\begin{proof}
(i)~and~(iii) are the displayed composition: $g_o^{-1}(s_o(x)) - b_o = w_o^{\top} z(x)$, so $h(s_o(x)) = g_n(c\,w_o^{\top} z(x) + b_n) = s_n(x)$; $h$ is strictly increasing for $c > 0$ and strictly decreasing for $c < 0$.
(ii)~Suppose $w_n$ is not a scalar multiple of $w_o$. Both cones $C_{\pm} = \{u : w_o^{\top} u > 0,\ \pm\, w_n^{\top} u > 0\}$ are nonempty and open: since $w_n \not\parallel w_o$ there is a $u_{\perp}$ with $w_o^{\top} u_{\perp} = 0$ and $w_n^{\top} u_{\perp} > 0$, and for any $u_0$ with $w_o^{\top} u_0 > 0$ the vector $u_0 + t\,u_{\perp}$ lies in $C_{+}$ and $u_0 - t\,u_{\perp}$ in $C_{-}$ for $t$ large. Richness supplies $z_i - z_j \in C_{+}$ and $z_k - z_l \in C_{-}$. By strict monotonicity of $g_o, g_n$, the first pair satisfies $s_o(x_i) > s_o(x_j)$ and $s_n(x_i) > s_n(x_j)$ (order preserved), the second $s_o(x_k) > s_o(x_l)$ and $s_n(x_k) < s_n(x_l)$ (order reversed). An increasing $h$ fails on the second pair, a decreasing $h$ on the first, and any monotone recalibration mis-ranks whichever pair contradicts its direction.
\end{proof}

\begin{corollary}[what a global map can and cannot absorb]
\label{cor:gap}
If $g_o = g_n = \mathrm{id}$, the version gap is affine in the assessments: $\Delta(x) = s_o(x) - s_n(x) = (w_o - w_n)^{\top} z(x) + (b_o - b_n)$. It is constant across objects iff $w_o = w_n$. An affine correction of old scores removes the constant part and the component of $w_o - w_n$ along $w_o$; any orthogonal component leaves object-level residual dispersion that no function of $s_o$ can absorb, because objects with equal $s_o$ carry different $\Delta$.
\end{corollary}

\textbf{Empirical correspondences.} The model's three degrees of freedom map onto the paper's measurements as follows. \emph{Origin:} the mean gap is repairable by fitting on representative objects (pool-refit explained fraction $0.76$--$0.81$ on Qwen; Table~\ref{tab:refit}), as any $b$-change is. \emph{Geometry:} the sub-$8\%$ dispersion absorption (Appendix~\ref{app:refit}), the tertile sign reversal (Gemma $+0.80$ vs.\ $-0.43$ SD; Appendix~\ref{app:pool}), and the excess top-30\% flips beyond the same-version noise floor ($+12$--$35$ points; Appendix~\ref{app:api4th}) are the three signatures Corollary~\ref{cor:gap} and Proposition~\ref{prop:collinear}(ii) predict for $w_n \not\parallel w_o$ and rule out for a pure origin change. \emph{Resolution:} a near-flat $g_n$ over the realized range produces near-constant scores regardless of $w_n^{\top} z$, while upstream channels remain intact, which is the measured Mistral~3.2 pattern (object-specific rationales, discriminative defect counts, saturated scalar map; Appendix~\ref{app:forensics}); under collapse, recalibrating $s_o$ is vacuous and the instrument fails outright.

\textbf{Scope.} Three idealizations bound the reading. First, $z(x)$ is latent: the data identify the class of change through the signatures above, not the weights themselves; estimating $\Delta w$ requires interventions that move one criterion at a time with matched placebo edits, which our design does not include. Second, linearity is assumed; a nonlinear aggregator changes the algebra but not the moral, since the proof only needs two objects whose relative order the versions disagree on. Third, $z(x)$ is taken version-invariant; the rationale coding of Appendix~\ref{app:structural} suggests the versions partly \emph{read} objects differently (content-shift majorities on two chains), a version-dependent $z_v(x)$ that global scalar maps can represent even less. The G7 nulls (compression and de-hedging leave the gap unchanged) narrow the changed directions to content properties that survive both edits, consistent with criterion-level rather than surface-level movement.

\end{document}